\documentclass[11pt]{article}

\newif\ifsoda
\sodafalse

\newif\ifnocref
\nocreffalse

\newif\ifsodacut
\sodacutfalse

\newif\ifauthorscut
\authorscutfalse

\usepackage{fullpage}

\usepackage{graphicx}

\usepackage{amsmath}
\usepackage{amsthm}
\usepackage{amssymb}
\usepackage{natbib}
\usepackage{mathtools}
\usepackage{hyperref}
\usepackage{bm}
\usepackage{subfig}

\usepackage{tikz}
\usepackage{pgfplots}
\usetikzlibrary{arrows}

\usepackage{dsfont}

\usepackage[noabbrev]{cleveref}

\usepackage{amsmath,amsfonts,amssymb,amsthm,latexsym}
\usepackage{graphicx}
\usepackage{color}
\usepackage{bm}
\newtheorem{theorem}{Theorem}
\newtheorem{definition}{Definition}
\newtheorem{lemma}{Lemma}
\newtheorem{proposition}{Proposition}
\newtheorem{fact}{Fact}
\newtheorem{corollary}{Corollary}

\newcommand{\vecbz}{\textbf{z}}
\newcommand{\scrlh}{g}
\newcommand{\scrbh}{G}

\newcommand{\qedend}{\qedhere}

\newcommand{\zbar}{\bar{z}_1}
\newcommand{\zees}{\vecbz_{-1}}
\newcommand{\optzees}{\vecbz_{-1}^*}
\newcommand{\optzeei}[1][i]{z_{#1}^*}

\newcommand{\zum}[1][\zbar]{{\cal S}(#1)}

\newenvironment{numberedtheorem}[1]{%
\renewcommand{\thetheorem}{#1}%
\begin{theorem}}{\end{theorem}\addtocounter{theorem}{-1}}

\newenvironment{numberedlemma}[1]{%
\renewcommand{\thelemma}{#1}%
\begin{lemma}}{\end{lemma}\addtocounter{lemma}{-1}}

\newenvironment{numberedcorollary}[1]{%
\renewcommand{\thecorollary}{#1}%
\begin{corollary}}{\end{corollary}\addtocounter{corollary}{-1}}

\newcommand{\agind}[1][i]{_{#1}}

\newcommand{\ironed}{\bar}
\newcommand{\constrained}{\hat}
\newcommand{\optconstrained}{\composed{\optimized}{\constrained}}
\newcommand{\optimized}[1]{#1\opt}
\newcommand{\differentiated}[1]{#1'}

\newcommand{\tagged}[2]{{#2}^{#1}}

\newcommand{\primedarg}[1]{#1\primed}
\newcommand{\noaccents}[1]{#1}
\newcommand{\composed}[3]{#1{#2{#3}}}

\newcommand{\newagentvar}[3][\noaccents]{%
\expandafter\newcommand\expandafter{\csname #2\endcsname}{#1{#3}}%
\expandafter\newcommand\expandafter{\csname #2s\endcsname}{#1{\boldsymbol{#3}}}%
\expandafter\newcommand\expandafter{\csname #2smi\endcsname}[1][i]{#1{\boldsymbol{#3}}_{-##1}}%
\expandafter\newcommand\expandafter{\csname #2i\endcsname}[1][i]{#1{#3}\agind[##1]}%
\expandafter\newcommand\expandafter{\csname #2ith\endcsname}[1][i]{#1{#3}_{(##1)}}%
}

\newcommand{\newitemvar}[3][\noaccents]{%
\expandafter\newcommand\expandafter{\csname #2\endcsname}{#1{#3}}%
\expandafter\newcommand\expandafter{\csname #2s\endcsname}{#1{\boldsymbol{#3}}}%
\expandafter\newcommand\expandafter{\csname #2smj\endcsname}[1][j]{#1{\boldsymbol{#3}}_{-##1}}%
\expandafter\newcommand\expandafter{\csname #2j\endcsname}[1][j]{#1{#3}_{##1}}%
\expandafter\newcommand\expandafter{\csname #2jth\endcsname}[1][j]{#1{#3}_{(##1)}}%
}

\newcommand{\maxval}{h}

\newcommand{\payalg}{\mathcal{A}}

\newcommand{\forrezs}[1]{{#1}^{\rezs}}
\newagentvar[\forrezs]{imbal}{\phi}
\newagentvar[\forrezs]{cumimbal}{\Phi}

\newagentvar{alloc}{x}

\newagentvar{falloc}{z}

\newagentvar{quant}{q}
\newagentvar[\constrained]{exquant}{\quant}
\newagentvar[\constrained]{critquant}{\quant}  
\newagentvar[\optconstrained]{monoq}{\quant}
\newagentvar{Val}{\nu}
\newagentvar{toquant}{Q}

\newagentvar{qprice}{\price}
\newagentvar{qrev}{R}
\newagentvar[\ironed]{iqrev}{\qrev}

\newagentvar[\bar]{bstrat}{\strat}
\newagentvar[\bar]{bpay}{\pay}
\newagentvar[\bar]{bwal}{\wal}
\newagentvar[\bar]{bbid}{\bid}
\newagentvar{zpay}{\pay^{\rezs}}
\newagentvar{zwal}{\wal^{\rezs}}
\newagentvar{zdelt}{\Delta^{\rezs}}
\newagentvar{bomega}{\bm{\omega}}

\newagentvar{qalloc}{y}
\newagentvar{cumalloc}{Y}
\newagentvar[\constrained]{calloc}{\qalloc}
\newagentvar[\constrained]{cumcalloc}{\cumalloc}
\newagentvar[\ironed]{ialloc}{\qalloc}
\newagentvar[\ironed]{icumalloc}{\cumalloc}

\newcommand{\exposted}[1]{#1^{\text{\it EP}}}
\newagentvar[\composed{\exposted}{\constrained}]{excalloc}{\qalloc}
\newagentvar[\exposted]{exalloc}{\qalloc}
\newagentvar[\exposted]{extalloc}{\talloc}
\newagentvar[\exposted]{exfalloc}{\falloc}

\newagentvar{rev}{R}
\newagentvar[\differentiated]{marg}{\rev}
\newagentvar{rawrev}{P}
\newagentvar[\differentiated]{rawmarg}{\rawrev}

\newagentvar[\primedarg]{inducedrev}{\rev}

\newagentvar[\tilde]{pseudorev}{\rev}
\newagentvar[\tilde]{pseudorawrev}{\rawrev}

\newagentvar{typespace}{{\cal T}}
\newagentvar{typesubspace}{S}

\newagentvar{type}{t}
\newagentvar{othertype}{s}
\newagentvar{val}{v}
\newagentvar{outcome}{w}
\newagentvar{outcomespace}{{\cal W}}

\newcommand{\served}[1]{#1^1}
\newcommand{\nonserved}[1]{#1^0}
\newcommand{\alloced}[1]{#1^{\alloc}}
\newcommand{\allocedi}[1]{#1^{\alloci}}
\newagentvar[\alloced]{xoutcome}{\outcome}
\newagentvar[\allocedi]{xioutcome}{\outcome}
\newagentvar[\served]{soutcome}{\outcome}

\newagentvar[\nonserved]{nsoutcome}{\outcome}

\newagentvar{price}{p}
\newagentvar{talloc}{\alloc}
\newagentvar[\tilde]{eval}{\val}
\newagentvar[\tilde]{ewal}{\wal}
\newagentvar[\tilde]{balloc}{\alloc}
\newagentvar[\tilde]{dalloc}{y}
\newagentvar[\tilde]{eballoc}{y}
\newagentvar{ealloc}{y}
\newagentvar[\tilde]{bprice}{\price}

\newagentvar{act}{a}
\newagentvar{bidspace}{A}
\newagentvar{actspace}{A}

\newagentvar[\constrained]{critval}{\val}
\newagentvar[\constrained]{crittype}{\type}
\newagentvar[\constrained]{critvirt}{\virt}
\newagentvar[\constrained]{reserve}{\val} 
\newagentvar[\constrained]{bidreserve}{\bid} 
\newagentvar[\optconstrained]{monop}{\val}
\newagentvar[\constrained]{monot}{\type}
\newagentvar[\optimized]{monorev}{\rev}

\newagentvar{rcalloc}{y}
\newagentvar[\optimized]{optrcalloc}{\rcalloc}
\newagentvar{biddist}{G}

\newagentvar[\constrained]{critbid}{\bid}
\newagentvar[\constrained]{cbid}{B}
\newagentvar[\primedarg]{wbid}{\bid}
\newagentvar{gfunc}{\vartheta}
\newagentvar{block}{C}

\newagentvar{util}{u}

\newagentvar{strat}{\beta}
\newagentvar{bid}{b}

\newagentvar{pay}{\pi}
\newagentvar{rez}{\rho}
\newagentvar[\tilde]{trez}{\rez}

\newagentvar[\forrezs]{minsum}{r}

\newagentvar{virt}{\phi}
\newagentvar{cumvirt}{\Phi}
\newagentvar{qvirt}{\phi}
\newagentvar[\ironed]{ivirt}{\virt}
\newagentvar[\ironed]{icumvirt}{\cumvirt}

\newagentvar[\tagged{\text{SD}}]{sdvirt}{\virt}
\newagentvar[\composed{\tagged{\text{SD}}}{\ironed}]{sdivirt}{\virt}
\newagentvar[\tagged{\text{MD}}]{mdvirt}{\virt}
\newagentvar[\composed{\tagged{\text{MD}}}{\ironed}]{mdivirt}{\virt}

\newagentvar{dist}{F}
\newagentvar{dens}{f}
\newagentvar{hazard}{h}
\newagentvar{cumhazard}{H}

\newagentvar[\ironed]{iprice}{\price}  
\newagentvar[\ironed]{ival}{\val}  

\newagentvar{ints}{{\cal I}}

\newagentvar{wal}{w}
\newagentvar{Util}{U}

\newitemvar{pos}{j}
\newitemvar{weight}{w}
\newitemvar[\differentiated]{mweight}{\weight}
\newitemvar{udtype}{\type}
\newitemvar[\constrained]{udcrittype}{\type}
\newitemvar{udalloc}{\alloc}
\newitemvar{udprice}{\price}
\newitemvar[\constrained]{udcalloc}{\qalloc}
\newitemvar[\constrained]{udcumcalloc}{\cumalloc}

\newagentvar{mech}{{\cal M}}
\newagentvar[\skew{5}{\hat}]{cmech}{{\cal M}}
\newagentvar{alg}{{\cal A}}

\newcommand{\reals}{{\mathbb R}}

\newagentvar{trans}{\sigma}

\newagentvar{demandset}{S}

\newcommand{\opt}{^{\star}}
\newcommand{\primed}{^\dagger}

\newagentvar{distout}{w}
\newagentvar{idistout}{\bar{w}}

\newagentvar{gap}{\delta}

\DeclareMathOperator{\argmin}{argmin}

\newcommand{\given}{\,\mid\,}

\newcommand{\prob}[2][]{\text{\bf Pr}\ifthenelse{\not\equal{}{#1}}{_{#1}}{}\!\left[{\def\givenn{\middle|}#2}\right]}
\newcommand{\expect}[2][]{\text{\bf E}\ifthenelse{\not\equal{}{#1}}{_{#1}}{}\!\left[{\def\givenn{\middle|}#2}\right]}

\newcommand{\tparen}{\big}
\newcommand{\tprob}[2][]{\text{\bf Pr}\ifthenelse{\not\equal{}{#1}}{_{#1}}{}\tparen[{\def\given{\tparen|}#2}\tparen]}
\newcommand{\texpect}[2][]{\text{\bf E}\ifthenelse{\not\equal{}{#1}}{_{#1}}{}\tparen[{\def\given{\tparen|}#2}\tparen]}

\newcommand{\sprob}[2][]{\text{\bf Pr}\ifthenelse{\not\equal{}{#1}}{_{#1}}{}[#2]}
\newcommand{\sexpect}[2][]{\text{\bf E}\ifthenelse{\not\equal{}{#1}}{_{#1}}{}[#2]}

\newcommand{\dd}{{\mathrm d}}

\newcommand{\pricelevelfigure}{%
\begin{tikzpicture}

\draw[->] (0,0) -- (10,0) node[anchor=north] {$s$};
\draw[->] (0,0) -- (0,4) node[anchor=east] {$\wali$};

\draw	(3,.1) -- (3,-.1) node[anchor=north] {$\minsumi$};

\draw	(.1,3.5) -- (-.1,3.5) node[anchor=east] {$\wali(h)$};

\draw[dotted] (0,3.5) -- (10,3.5);

\draw	(.1,.2) -- (-.1,.2) node[anchor=east] {$\wali(0)$};

\draw[dotted] (0,.2) -- (10,.2);

\draw[dotted] (0,0) -- (4,4) node[anchor=south] {$s$};

\draw[dotted] (0,0) -- (8,4) node[anchor=south] {$s/2$};

\draw[line width=2pt,rounded corners,dashed,color=gray] (10,.4) -- (7,.45) -- (6,.5) -- (5.5,1) -- (3.5, 1.2) -- (3,1.5) -- (2.5,1.7) -- (2.7,2) 
   -- (3,2.5) -- (4,2.7) -- (3.8, 3.2) -- (4.2,3.5) node[anchor=south west,color=black] {$\mathcal{Q}^{\rezs}_i$};

\draw[thick,rounded corners,] (10,.4) node[anchor=south] {$\mathcal{P}^{\rezs}_i$} -- (7,.45) -- (6,.5) -- (5.5,1) -- (3.5, 1.2) -- (3,1.5);

\end{tikzpicture}}

\newcommand{\pricelevelcasesfigure}{%
\begin{tikzpicture}

\draw[->] (0,0) -- (10,0) node[anchor=north] {$s$};
\draw[->] (0,0) -- (0,4) node[anchor=east] {$\wali$};

\draw	(3,.1) -- (3,-.1) node[anchor=north] {$\minsumi$};

\draw	(8,.1) -- (8,-.1) node[anchor=north] {$\minsumi{{}^{{}'}}$};

\draw	(.1,3.5) -- (-.1,3.5) node[anchor=east] {$\wali(h)$};

\draw[dotted] (0,3.5) -- (10,3.5);

\draw	(.1,.2) -- (-.1,.2) node[anchor=east] {$\wali(0)$};

\draw[dotted] (0,.2) -- (10,.2);

\draw[dotted] (0,0) -- (4,4) node[anchor=south] {$s$};

\draw[dotted] (0,0) -- (8,4) node[anchor=south] {$s/2$};

\draw[line width=2pt,rounded corners,dashed,color=gray] (.2,.2) -- (10,.2);

\draw[thick,rounded corners] (.4,.2) -- (10,.2);

\draw[line width=2pt,rounded corners,dashed,color=gray] (10,.4) -- (7,.45) -- (6,.5) -- (5.5,1) -- (3.5, 1.2) -- (3,1.5) -- (2.5,1.7) -- (2.7,2) 
   -- (3,2.5) -- (4,2.7) -- (3.8, 3.2) -- (4.2,3.5) node[anchor=south west,color=black] {$\mathcal{Q}^{\rezs}_i$};

\draw [fill] (4.2,3.5) circle [radius=0.05];

\draw[thick,rounded corners,] (10,.4) node[anchor=south] {$\mathcal{P}^{\rezs}_i$} -- (7,.45) -- (6,.5) -- (5.5,1) -- (3.5, 1.2) -- (3,1.5);

\draw[line width=2pt,rounded corners,dashed,color=gray] (10,2) -- (9,2.2) -- (8.7,2.5) -- (8.5,3) -- (8.1,3.2) -- (8,3.5) node[anchor=south west,color=black] {$\mathcal{Q}^{\rezs'}_i$};

\draw [fill] (8,3.5) circle [radius=0.05];

\draw[thick,rounded corners] (10,2) node[anchor=south] {$\mathcal{P}^{\rezs'}_i$} -- (9,2.2) -- (8.7,2.5) -- (8.5,3) -- (8.1,3.2) -- (8,3.5);

\end{tikzpicture}}

\newcommand{\multipleinversesfigure}{%
\begin{tikzpicture}
  \begin{axis}[xlabel={$\alpha$},
      xmin=0,xmax=10,
      scaled ticks=false,
      y label style={at={(axis description cs:-.05,.5)},anchor=south},
      y tick label style={/pgf/number format/precision=3,/pgf/number format/fixed},
    ylabel={$\payi[1](\alpha)-\payi[5](\alpha)$}]
\pgfplotstableread{ca_data.txt}{\cadata}
\addplot[color=black,mark=none,thick] table from \cadata;
\addplot[color=black,mark=none,dotted] coordinates { (0,0) (10,0) };
\end{axis}
\end{tikzpicture}}

\begin{document}



\begin{titlepage}

\title{Inference from Auction Prices\footnote{Arxiv: https://arxiv.org/abs/1902.06908}}

\newcommand{\email}[1]{\href{mailto:#1}{#1}}

\author{Jason Hartline\thanks{Northwestern U., Evanston IL.  Work done in part while supported by NSF CCF 1618502.\newline Email: \email{hartline@northwestern.edu}} \and Aleck Johnsen\thanks{Northwestern U., Evanston IL.  Work done in part while supported by NSF CCF 1618502. \newline Email: \email{aleckjohnsen@u.northwestern.edu}} \and Denis Nekipelov\thanks{U. of Virginia, Charlottesville, VA.  Work done in part while supported by NSF CCF 1563708.\newline  Email: \email{denis@virginia.edu}}  \and Zihe Wang\thanks{Shanghai University of Finance and Economics (SUFE). Work done in part while supported by the Shanghai Sailing Program (Grant No. 18YF1407900), the National NSFC Grant 61806121, Innovation Program of Shanghai Municipal Education Commission, Program for Innovative Research Team of Shanghai University of Finance and Economics, and the Fundamental Research Funds for the Central Universities. Email: \email{wang.zihe@mail.shufe.edu.cn}}}



\maketitle

\begin{abstract}
Econometric inference allows an analyst to back out the values of agents in a mechanism from the rules of the mechanism and bids of the agents.  This paper gives an algorithm to solve the problem of inferring the values of agents in a dominant-strategy mechanism from the social choice function implemented by the mechanism and the per-unit prices paid by the agents (the agent bids are not observed).  For single-dimensional agents, this inference problem is a multi-dimensional inversion of the payment identity and is feasible only if the payment identity is uniquely invertible.  The inversion is unique for single-unit proportional weights social choice functions (common, for example, in bandwidth allocation); and its inverse can be found efficiently.  This inversion is not unique for social choice functions that exhibit complementarities.  Of independent interest, we extend a result of \citet{ros-65}, that the Nash equilbria of ``concave games'' are unique and pure, to an alternative notion of concavity based on \citet{GN-65}.
\end{abstract}

\end{titlepage}

\section{Introduction}
\label{s:intro}

Traditional econometric inference allows an analyst to determine the
values of agents from their equilibrium actions and the rules of a
mechanism \citep{GPV-00,HT-03}.  This paper studies an inference
problem when only the profile of the agents' per-unit prices is
available to the analyst.  Such an inference may be applicable when
bids are kept private but prices are published; moreover, it is of
interest even for incentive compatible mechanisms (where agents
truthfully report their preferences).  As a motivating example, with
the per-unit prices from the incentive compatible mechanism for
allocating a divisible item proportionally to agent values
\citep[cf.][]{JT-04}, we prove that agents' values are uniquely
determined and can be computed efficiently.

Econometric inference is a fundamental topic in a data-driven approach
to mechanism design and a number of recent papers have been developing
its algorithmic foundations.  The following are prominent examples.
\citet{CHN-14,CHN-16} show that the revenue and welfare of a counter
factual auction can be estimated directly from Bayes-Nash equilibrium
bids in an incumbent auction.  \citet{NST-15} develop methods for
identifying the rationalizable set of agent values and regret
parameters in repeated auctions with learning agents. \citet{HNS-17}
show that the quantities that govern price-of-anarchy analyses can be
determined directly from bid data and, thus, empirical
price-of-anarchy bounds can be established that improve on the
theoretical worst case.  

There are two important questions in algorithmic econometrics.  First,
when are the values uniquely identified?  Second, can the values
be efficiently computed when the values are identifiable?  The first
question is studied in depth by the econometrics literature (for
inference from actions); the second question is an opportunity for
algorithms design and analysis.

We consider inference in single-dimensional environments where a
stochastic social choice function maps profiles of agent values to
profiles of allocation probabilities.  The characterization of
incentive compatibility \citep{mye-81} requires the allocation
probability of an agent be monotonically non-decreasing in that
agent's value and that an agent's expected payments satisfy a {\em
  payment identity}.  Per-unit prices -- the expected payments
conditioned on winning -- are easily determined from the expected
payments in the payment identity by normalizing by the allocation
probability.\footnote{Our methods are written assuming that per-unit
  prices are observed rather than expected payments.  These prices are
  more natural for mechanisms usually considered in algorithmic
  mechanism design as they arise in mechanisms where losers pay
  nothing, i.e., ex post individually rational mechanism.  If instead
  the realized expected payments and realized allocation
  probabilities are observed, then these per-unit prices can be easily
  calculated and our methods applied to the result.}  Consequentially,
given any social choice function and valuation profile, the allocation
probabilities and prices of an incentive compatible mechanism that
implements the social choice function are uniquely and easily
determined.  Our inference problem is the opposite.  Given the profile
of the agents' prices, determine the valuation profile that leads to
these prices.  The social choice function and, thus, the function
mapping values to prices is known.  The resulting inversion problem is
multi-dimensional and this multi-dimensionality leads to a possibility
of non-uniqueness (and consequentially, non-identifiability) and
computational challenges.

The first goal of this paper is to understand what social choice
functions admit inference from prices and which do not.
Fundamentally, social choice functions with induced allocation rules
that are not strictly increasing do not admit inference.  For example, the only inference possible from the outcome of a second-price auction is
that the winner has value above the winner's price and the losers have
value below the winner's price.  On the other hand, a ``soft max''
social choice function like proportional values, where an agent
receives a fraction of the item proportional to her value, is strictly
continuous and, as we will show, the valuation profile can be uniquely
inferred from the winner-pays prices.  We will show sufficiency for social choice
functions to admit inference from prices as ones where the
Jacobian of the payment identity has all minors positive on (almost) all inputs and, as a class, proportional weights social choice functions (with
general strictly monotonic weight functions) satisfy this property.
In contrast we show that this property does not generally hold for
social choice functions that exhibit complementarities.

These identification and non-identification results are complemented
by an algorithm for efficiently computing the valuation profile from
the prices that corresponds to any proportional weights social choice
function for single-item environments.

Our focus is on proportional weights allocation rules for
(probabilistically) sharing a unit resource.  Such mechanisms have
been previously considered in the literature on bandwidth allocation
\citep[e.g.,][]{JT-04}. Another point of contact with the literature
is the special case of exponential weights.  The mechanism that
implements the exponential weights allocation rule is known as the
exponential mechanism \citep{HK-12}.  The exponential mechanism is
often considered because its realized allocation has good privacy
properties.  \citet{HK-12} recommend additionally adding Laplacian
noise to the payments of the exponential mechanism so that its
realized outcome (allocation and payments) is differentially private.
Our main result shows that, in fact, without such noise added to the
payments the exponential mechanism is not private.

\paragraph{Organization.} 
The rest of this paper is organized as follows.
Section~\ref{s:prelim} gives notation for discussing social choice
functions, mechanisms, and agents; reviews the characterization of
incentive-compatible single-dimensional mechanisms; and reviews
proportional weights allocations.  \Cref{s:nasheq}, then, gives an
algorithmic framework for robustly identifying values from prices.  It
shows that values are identified from payments corresponding to social
choice functions given by proportional weights in single-item and
multi-unit environments.  \Cref{s:hardness} shows that values are not
identifiable from prices for proportional weights allocations that
correspond to environments with
complementarities.  \Cref{s:computation} gives an efficient algorithm
for inferring values from prices for proportional weights social choice
functions in single-item  environments.

\section{Preliminaries}
\label{s:prelim}

This paper considers general environments for single-dimensional
linear agents.  Each agent $i$ has value
$\vali\in\left[0,\maxval\right]$.  For allocation probability $\alloci$
and expected payment $\pricei$, the agent's utility is $\vali\,\alloci
- \pricei$.  A profile of $n$ agent values is denoted $\vals =
(\vali[1],\ldots,\vali[n])$; the profile with agent $i$'s value
replaced with $z$ is $(z,\valsmi) =
(\vali[1],\ldots,\vali[i-1],z,\vali[i+1],\ldots,\vali[n])$.

A {\em stochastic social choice function} $\allocs$ maps a profile of
values $\vals$ to a profile of allocation probabilities.  A {\em
  dominant strategy incentive compatible (DSIC)} mechanism
$(\allocs,\prices)$ maps a profile of values $\vals$ to profiles of
allocations $\allocs(\vals)$ and payments $\prices(\vals)$ so that:
for all agents $i$, values $\vali$, and other agent values $\valsmi$,
it is optimal for agent $i$ to bid her value $\vali$.  The following
theorem of \citet{mye-81} characterizes social choice functions that
can be implemented by DSIC mechanisms.

\begin{theorem}[\citealp{mye-81}]
\label{thm:myerson}
Allocation and payment rules $(\allocs,\prices)$ are induced by a
dominant strategy incentive compatible mechanism if and only if for
each agent $i$,
\begin{enumerate}
\item (monotonicity) 
\label{thmpart:monotone}
allocation rule $\alloci(\vali,\valsmi)$ is monotone non-decreasing in
$\vali$, and
\item 
\label{thmpart:payment}
(payment identity) payment rule $\pricei(\vals)$ satisfies
\begin{align}
\label{eq:payment-identity}
\pricei(\vals) &= \vali\, \alloci(\vals) - \int_0^{\vali}
\alloci(z,\valsmi)\, \dd z + \pricei(0,\valsmi),
\end{align}
\end{enumerate}
where the payment of an agent with value zero is often
zero, i.e., $\pricei(0,\valsmi) = 0$.
\end{theorem}

Most DSIC mechanisms are implemented to satisfy an ex post individual
rationality constraint; specifically, an agent pays nothing when not allocated.  The payment when allocated, i.e., the
{\em per-unit price}, is thus the expected payment normalized by the probability of winning.  Throughout this work, we assume $\pricei(0,\valsmi) = 0$.  Denote the {\em price function} by $\pays : \reals^n_+ \to
\reals^n_+$, as
\begin{align}
\nonumber
\payi(\vals) &= \pricei(\vals)/\alloci(\vals)\\
\label{eq:winner-pays-bid-strats}
             &= \vali - \frac{\int_0^{\vali} \alloci(z,\valsmi)\,\dd z}{\alloci(\vals)} 
\end{align}
for all agents $i$. 

The main objective of this paper is to infer the agents' values from
observations of the per-unit prices of the mechanism.  A price profile
$\rezs$ is observed, and it is desired to infer the
valuation profile $\vals$ that generated this price profile by $\rezs = \pays(\vals)$.  The key
question of this paper is to identify sufficient conditions on the
social choice function $\allocs$ such that the price
function $\pays$ is invertible.

An important special case is the case where there is $n=1$ agent and
the price function $\pay(\cdot)$ is single-dimensional.  When the
social choice function $\alloc(\cdot)$ is strictly increasing, the
price function $\pay(\cdot)$ is strictly increasing (apply \Cref{lem:monoprice} with only one agent), and is uniquely invertible.  Thus, the
agent's value can be identified from her observed price $\rez$, e.g.,
by binary search.

\begin{lemma}
\label{lem:monoprice}
Assume $\partial\alloci/\partial\vali(\vals)>0$ everywhere.  Then $\partial\payi/\partial\vali(\vals)>0$ for all values except 0.
\end{lemma}

\begin{proof}
The partial of the price function $\payi'(\vali,\valsmi) =
  \frac{\alloci'(\vali,\valsmi)\int_0^{\vali}
    \alloci(z,\valsmi)dz}{\left(\alloci(\vali,\valsmi)\right)^2}$ is positive if
  $\alloci'(\vali,\valsmi)$ is positive, unless the numerator is 0 because $\vali=0$ and the integral endpoints are the same.
\end{proof}

Our goal is to understand families of (multi-agent) social choice
functions $\allocs$ that allow values to be inferred from prices.
Clearly, as in the single-agent case, if the allocation rule is not
strictly increasing in each agent's value, then the values of the
agents cannot be inferred.  We assume that the social choice function
$\allocs$ is such that it has strictly-increasing allocation functions
$\alloci$ for any given $\valsmi$, for all $\vali>0$.

Mechanisms
in the literature for welfare and revenue maximization are based on
social choice functions that map agents' values to weights and
allocate to maximize the sum of the weights of the agents allocated.
In order to satisfy the required strict monotonicity property, our
focus is on smoothed versions of these social choice functions under
feasibility constraints that correspond to single-item auctions (or
single-minded combinatorial auctions admitting only one winner).

In single-item environments a natural ``soft max'' is given by
proportional weights allocations.  A weight function is given for each
agent $i$ as a strictly monotone and continuously differentiable
function $\wali : \reals_+ \to \reals_+$ and the proportional weights
social choice function maps each agent's value to a weight and then
allocates to agents with probabilities proportional to
weights.\footnote{For simplicity, we
  assume that all weights functions are everywhere strictly positive
  for all agents, even at $\vali = 0$.}  A canonical example of
proportional weights is exponential weights: $\wali(\vali) =
e^{\vali}$ for each agent $i$.

Given the assumptions on functions $\wals$, they are invertible. Where
appropriate we will overload $\vali$ to allow it to be the functional
inverse of $\wali$ mapping a weight back to its value.  We also
overload the notations $\allocs, \pays$ to take weights $\wals$ as an
input, with $\allocs(\wals)\coloneqq \allocs(\vals(\wals))$ and
$\pays(\wals)\coloneqq\pays(\vals(\wals))$.

\section{Identification and Non-identification}
\label{s:nasheq}

This section considers sufficient conditions under which values can be inferred from the observed prices $\rezs$ of a DSIC mechanism $(\allocs,\prices)$.  
The critical challenge to identification arises from the observation that values can only possibly be identified from prices if the price function $\pays$ is invertible.  
We solve this challenge both in theory here in \Cref{s:nasheq}, and algorithmically in \Cref{s:computation}.
  
Our theoretical and algorithmic results are simpler to prove as inversions from prices to intermediate weights, and then from weights to values.  Describing the inversion via weights is without loss because weights functions $\wali(\cdot)$ are continuously differentiable, positive, strictly increasing functions mapping an agent's value to weight.  The weights can be inverted as $\vali(\wali) \vcentcolon= \wali^{-1}(\vali)$.  
 

Our approach is to write the problem of inverting
the price function $\pays$ at prices $\rezs$ as a proxy game between
proxy players where the actions are weights.  The proxy game is a tool for computing the inverse: {\em with proxy actions corresponding to weights, its unique Nash will be the desired inversion point}.  Each proxy player represents an agent of the mechanism $(\allocs,\prices)$.  A proxy player $i$ is responsible for identifying its agent's weight $\wali$, in the proxy game parameterized by $\rezs$.  

Towards designing the proxy game to have a specific (and unique) Nash equilibrium, we design the proxy game's payoff function $\cumimbali$ of a proxy player $i$ for (action) $\ewali$ -- given the profile of weight-actions from the other proxy players $\ewalsmi$ -- to be 
optimized where $\payi(\ewals)$ on the proxy action profile is closest to the observed price $\rezi$.\footnote{The full importance of the proxy game construction is realized when ``erroneous" prices are used as inputs, as the proxy game is still defined with action space corresponding to weights space, is still continuous, and will still have a unique pure Nash equilibrium which can be output.}  The first goal here is to give the technical description a price-function inversion algorithm using a proxy game, and reduce the question of its correctness to the uniqueness of a pure Nash equilibrium in the proxy game (\Cref{prop:nash-implementation}).

Recalling
equation~\eqref{eq:winner-pays-bid-strats}, we transform the price
function $\pays$ to weights-space using calculus-change-of-variables
as
	\begin{equation}
    \label{eqn:bidfn:w}
	\payi(\wals)=\vali(\wali)-\frac{\int^{\wali}_{\wali(0)}\alloci(z,\walsmi)\vali'(z)dz}{
    \alloci(\wals)}.
	\end{equation}

\noindent For fixed
observed prices $\rezs$, define the {\em price-imbalance function} $\imbali(\cdot)$ and the
{\em cumulative price-imbalance} $\cumimbali(\cdot)$ respectively as follows, and we set $\cumimbals$ as the proxy game utility function:
\ifsoda
\begin{align}
\label{eqn:overpay}
&\imbali(\ewali,\ewalsmi) = \rezi - \payi(\ewali,\ewalsmi)\\
\notag
&\quad= \rezi - \vali(\ewali) + \frac{\int^{\ewali}_{\wali(0)}
  \alloci(z,\ewalsmi)\vali'(z)\,\dd z}{\alloci(\ewals)},\\
\label{eqn:computil}
&\cumimbali(\ewali,\ewalsmi) = \int_{\wali(0)}^{\ewali} \imbali(z,\ewalsmi)\,\dd z.
\end{align}
\else
\begin{align}
\label{eqn:overpay}
\imbali(\ewali,\ewalsmi) &= \rezi - \payi(\ewali,\ewalsmi) = \rezi - \vali(\ewali) + \frac{\int^{\ewali}_{\wali(0)}
  \alloci(z,\ewalsmi)\vali'(z)\,\dd z}{\alloci(\ewals)},\\
\label{eqn:computil}
\cumimbali(\ewali,\ewalsmi) &= \int_{\wali(0)}^{\ewali} \imbali(z,\ewalsmi)\,\dd z.
\end{align}
\fi

\noindent The proxy game is defined with weights $\ewals$ as proxy actions, and with utilities for the proxy agents given by the cumulative price-imbalance functions $\cumimbals$.  Each function $\cumimbali$ is strictly concave in dimension $i$, except at the lower end point of its domain where it is weakly concave (see \Cref{lem:concutil} in \Cref{a:pderi}).  From concavity of $\cumimbali$ in~\eqref{eqn:computil}, a ``zero" of $\imbali$ in~\eqref{eqn:overpay} is optimal.  As desired, when other players select proxy weights $\ewalsmi$, proxy player $i$ would select proxy weight $\ewali$ so that agent $i$'s price according to $\pays$ on $\ewals$ is closest to agent $i$'s observed payment $\rezi$ (and $\rezi=\payi(\ewals)$ if possible).  Based on this proxy game, we define the following inference algorithm.

\begin{definition} 
\label{d:nash-mech}
\label{d:proxy-game}
The {\em price-inversion algorithm} $\payalg$ on price space $[0,\infty)^n$ for social choice function
$\allocs$ on value space $[0,h]^n$ is
\begin{enumerate}
\item Observe price profile $\rezs$.
\item \label{step:nash} Select a Nash equilibrium $\ewals$ in the {\em proxy game} 
 defined in weight space 
 with utility functions given by the cumulative price-imbalance $\cumimbals$ for $\rezs$.
\item Return inferred values based on inferred weights $\ewals$ as $(\vali[1](\ewali[1]),\ldots, \vali[n](\ewali[n]))$.
\end{enumerate}
\end{definition}

A key property for the proper working of the price-inversion algorithm
is whether the proxy game admits a unique pure Nash equilibrium.  For
example, if there are multiple distinct valuation profiles that map to
the same prices via $\pays$ (values of agents in the original auction), then each of these valuation profiles will have a corresponding equilibrium in the proxy game (in proxy game action-weights space).  \Cref{prop:nash-implementation} formalizes the correctness of the price-inversion algorithm, subject to the proxy game having unique pure Nash equilibrium.  

\begin{proposition}
\label{prop:nash-implementation}
Any weights profile $\wals \in [\wali(0),\wali(\maxval)]^n$ such that observed price profile $\rezs$
satisfies $\rezs = \pays(\wals)$ is a Nash equilibrium of the proxy
game on the social choice function $\allocs$ and prices $\rezs$; if
this Nash equilibrium $\wals$ of the proxy game is unique then
the inverse $\pays^{-1}(\rezs)$ is unique and given by the price
inversion algorithm $\payalg$.
\end{proposition}

\begin{proof} The second part follows from the first part.  For the first part, assume $\rezs = \pays(\wals)$ for some $\wals$ in weights space domain.  Action profile $\wals$ in the proxy game is a Nash equilibrium as follows.  Each proxy agent's first-order condition is satisfied.  Specifically, with utilities given by the cumulative imbalances $\cumimbals$, the first-order condition is given by 
$\imbali(\wali,\walsmi) = \rezi - \payi(\wali,\walsmi)$ and is zero by
  the choice of $\wals$.  Further, checking first-order conditions is
  sufficient because $\cumimbals$ is strictly concave by \Cref{lem:concutil}, i.e.,
  $\frac{\partial\imbali(\wali,\walsmi)}{\partial\wali} = -\payi'(\wals) < 0$ (except at the lower bound where the partial is 0, but this can not affect player $i$'s strict preference over actions).
\end{proof}

Motivated by \Cref{prop:nash-implementation}, the remainder of this
section identifies proportional weights as a large natural class of
social choice functions for which the proxy game has a unique pure
Nash equilibrium for all price profiles, which we will state in \Cref{thm:themainresult}.  The computational question
of finding the Nash equilibrium of the proxy game is deferred to
\Cref{s:computation}.

We outline the rest of the section.  As mentioned previously, a necessary
condition for the uniqueness of pure Nash in the proxy game is that
the price function $\pays$ is one-to-one.  In \Cref{s:psdunique}, we
show that $\pays$ being one-to-one 
is implied by \textit{a slightly weaker}
condition than the following: for all inputs the Jacobian of $\pays$ -- denoted $J_{\pays}$ -- has all positive principal minors (i.e. it is a $P$-matrix, see \Cref{def:pmat} below).  In \Cref{s:oneweightspsd} we show that all proportional weights social choice
functions (for single-unit environments) induce price functions that
satisfy this condition. 
In contrast, \Cref{s:hardness}
describes a natural variant of proportional weights social choice
functions for environments which resemble single-minded combinatorial
auctions, and shows that the price functions for these social choice
functions are not generally invertible, and therefore the proxy game
does not have a unique pure Nash equilibrium in this extended setting.

\subsection{Sufficiency of ``Interior P-Matrix Functions"}
\label{s:psdunique}

This section shows that a sufficient condition for the uniqueness of a
pure Nash equilibrium in the proxy game defined in algorithm $\mathcal{A}$ (\Cref{d:nash-mech}) -- necessary for its correctness -- is that the price function $\pays$ (for the social choice rule $\allocs$) is an ``interior $P$-matrix function," a property on its Jacobian $J_{\pays}$ (to be defined shortly in \Cref{def:interiorpd}).  An intuitive outline of the technique is:
\begin{itemize}
    \item existence is by algorithm design, as the vector of true weights exists as a pure Nash point, in particular one with all first-order conditions equal to 0;
    \item uniqueness results because the mapping between proxy game action vectors and proxy agent utility gradients is a bijection with ``high-dimensional monotonicity," for which interior $P$-matrix functions are sufficient; so the proxy game has ``high-dimensional concavity."
\end{itemize}

\noindent We will address existence in 
\Cref{thm:gameunique} and its proof.  First we set up the structure towards uniqueness (also \Cref{thm:gameunique}).  The next definition for $P$-matrix (``positive matrix") comes from \citet{GN-65}, and so does \Cref{thm:psdunique} (below) connecting $P$-matrices to bijection and invertibility.\footnote{Further supporting results given in \Cref{a:gnext} are also from \cite{GN-65}.}   
We give their definition and extend it to include ``weak" and ``negative" cases, and list facts about $P$-matrices to be used in this and subsequent sections:
\begin{definition}
\label{def:pmat} A $K\times K$ matrix is a {\em $P$-matrix} if all of its principal minors are positive (i.e., have strictly positive determinant).  Such a matrix is a {\em $P_0$-matrix} if all of its principal minors are non-negative.  Further, the terms {\em $N$-matrix} and {\em $N_0$-matrix} are used to describe matrices that when negated (all entries multiplied by $-1$) are, respectively, a $P$-matrix and a $P_0$-matrix.
\end{definition}

\begin{fact}
\label{f:pmatfact}
The following are true about $P$-matrices:
\begin{enumerate}
    \item a $P$-matrix is {\em downward-closed}, i.e., each of its principal minors is a $P$-matrix too;
    \item the class of $P$-matrices contains the class of all positive definite matrices as a special case (where for our purposes, the definition of a positive definite matrix $M$ is $\vecbz^{\top} M \vecbz > 0~\forall~\vecbz\neq \mathbf{0}$ with $M$ not necessarily symmetric);
    \item the product of a strictly positive, diagonal matrix and a $P$-matrix is also a $P$-matrix.
\end{enumerate}
\end{fact}


\begin{theorem}[\citealp{GN-65}]
\label{thm:psdunique}
A continuously differentiable function $f : \Omega \to
\reals^n$ with compact and convex product domain $\Omega\subset\reals^n$ is one-to-one if its Jacobian is everywhere a $P$-matrix.
\end{theorem}

\noindent We will need a generalization of \Cref{thm:psdunique} that relaxes the strict $P$-matrix condition on the Jacobian, on the axis-aligned boundaries.  The problem for our price-function setting is that the pseudogradient of the utility function is only a $P_0$-matrix on the lower boundaries (from equation~\eqref{eqn:partialsame} below).

Define a function $f:\reals^n\rightarrow\reals^n$ to be a {\em $P$-matrix function} if its Jacobian is a $P$-matrix at all points of the function's domain.  We need to extend this definition.  
Note, \Cref{def:interiorpd} for {\em interior $P$-matrix functions} depends on \Cref{def:identbound} for {\em identified boundaries} (next).

\begin{definition}\footnote{We make frequent use of input space $\Omega$ in this paper as a compact and convex product space.  Unless noted specifically otherwise, we let the dimension-wise ranges be $\Omega=\left[a_1,b_1\right]\times \left[a_2,b_2\right]\times \cdots \times \left[a_n,b_n\right]$ as in this definition.}
\label{def:interiorpd}
For product space $\Omega = \left[a_1,b_1\right]\times \left[a_2,b_2\right]\times \cdots \times \left[a_n,b_n\right]$ and function $f:\Omega\to\reals^n$, a function $f: \Omega\rightarrow \reals^n$ 
is an {\em interior $P$-matrix function (respectively interior $N$-matrix function)} if for every point $\bomega\in\Omega$:

\begin{itemize}
\item the Jacobian of $f$ evaluated at $\bomega$ as $J_f(\bomega)$ is a $P_0$-matrix (respectively $N_0$-matrix);
\item and choosing the minor of $J_f(\bomega)$ that removes row/column pairs corresponding to the dimensions in which $\bomega$ is in identified boundaries of $\Omega$, this principal minor of $J_f(\bomega)$ is strictly a $P$-matrix (respectively $N$-matrix).
\end{itemize}
\end{definition}

\noindent Before giving the definition of {\em identified boundaries}, we give their intuition and justification.  They describe conditions which address the problem of Jacobians having determinant 0 at the boundaries.  By \Cref{thm:psdunique}, Jacobian as $P$-matrix everywhere is sufficient for inversion.  An identified boundary (in input space dimension $i$) must first allow unilateral inversion of its coordinate, by mapping to a unique, constant output in dimension $i$ for all inputs in this identified boundary of $i$ (i.e., bijectively); and second, after fixing the input in all such identified-boundary dimensions $i$ as parameters, the reduced function in the remaining dimensions must have Jacobian as a strict $P$-matrix, which will imply that it can be inverted; hence entire output vectors can be inverted.



\begin{definition}
\label{def:identbound}
For compact and convex product space $\Omega\subset\reals^n$, and function $f:\Omega\rightarrow \reals^n$, a {\em boundary} (described by $c_i\in\left\{a_i,b_i\right\}$) is {\em identified} if both of the following hold for all $\bomega_{-i} \in \Omega_{-i}$:
\begin{itemize}
    \item fixing $\bomegai=c_i$, function $f_i(c_i,\bomegasmi)$ is constant for all $\bomegasmi$; or equivalently, all cross-partials on the $c_i$ boundary are 0: $\partial f_i / \partial \omega_j (c_i,\bomega_{-i})=0$ for all $j\neq i$;
    \item the output is unique to the boundary: $f_i(c_i,\bomegasmi)\neq f_i(d_i,\bomegasmi)$ for all $d_i\in\left[a_i,b_i\right], d_i\neq c_i$.
\end{itemize}
\end{definition}

\noindent As previously suggested, the implication of an identified boundary is that, (e.g.) the low point of the domain in dimension $i$ maps identically to the low point of the function's range in dimension $i$ as a unilateral bijection.  Further, note that a sufficient condition for the second point of the definition is having partial $\partial f_i/\partial \omega_i>0$ for all inputs $\omega_i$ off the boundary.


Our \Cref{thm:gnext} generalizes \Cref{thm:psdunique} of~\citeauthor{GN-65}.  
We use it as an interim result towards our more pertinent result in \Cref{thm:gameunique}, which connects interior $P$-matrix functions to proxy games.  \Cref{thm:gameunique} lets us reduce the correctness of {\em price-inversion algorithm} $\payalg$ to the condition that $\pays$ is an interior $P$-matrix function, stated formally in \Cref{cor:nash-implementation}.  Proofs for the next two theorems are given in \Cref{a:gnext}.

\begin{theorem}
\label{thm:gnext}
If function $f:\Omega\to\reals^n$ on compact and convex product domain $\Omega\subset\reals^n$ is an interior $P$-matrix function (\Cref{def:interiorpd}), then it is one-to-one, and therefore invertible on its image.
\end{theorem}

\begin{theorem}
\label{thm:gameunique}
A game with $n$ players and
\begin{itemize}
\item a compact and convex product action space $\Omega_1\times\ldots\times\Omega_n=\Omega\subset\reals^n$; \item a continuous and twice-differentiable utility function $\Utils:\Omega\rightarrow\reals^n$ such that:
\begin{itemize}
\item the pseudogradient $\left[\frac{\partial U_i}{\partial \omega_i}\right]_i$ of the utility function $\Utils$ is an interior $N$-matrix function;
\item and there exists $\bomega^0\in\Omega$ such that the pseudogradient evaluated at $\bomega^0$ is $\mathbf{0}$ (the 0-vector);
\end{itemize}
\end{itemize}
has a unique Nash equilibrium, which is $\bomega^0$, and this equilibrium is pure.
\end{theorem}







\begin{corollary}
\label{cor:nash-implementation}
Given agents with (unknown) values $\vals\in\left[0,\maxval\right]^n$.  Consider price function $\pays$ resulting from a 
dominant-strategy incentive-compatible mechanism implementing $\allocs$, with Jacobian $J_{\pays}$.

If $\pays$ is an interior $P$-matrix function, 
then on observed prices from restricted domain $\rezs \in \text{Image}(\pays)$, the price-inversion algorithm $\payalg$ (\Cref{d:nash-mech}) infers successively the true weights $\wals$ and the true values $\vals$ from the mechanism's outcome (as summarized by the prices $\rezs=\pays(\vals)$).
\end{corollary}

\begin{proof}
We show that under the given assumptions, the proxy game meets the conditions of \Cref{thm:gameunique}.  The action space of the proxy game is equal to the agents' weights space which is a compact and convex product space.  The proxy game has payoffs given by $\cumimbals$ such that utility functions are continuous and twice-differentiable.

The pseudogradient of the payoffs is given by $\imbals$, and the Jacobian of the pseudogradient is the negation of the matrix $J_{\pays}$.  Given $\pays$ as an interior $P$-matrix function, its negation $-\pays$ is an interior $N$-matrix function.  The true values $\wals$ as input-actions to the proxy game will result in evaluation of the pseudogradient as $\imbals(\wals)=\mathbf{0}$ by design of the game, so $\bomega^0=\wals$ exists.

In conclusion, the proxy game indeed satisfies the conditions of \Cref{thm:gameunique}, and admits $\wals$ as a unique Nash equilibrium which is pure.  Defining the inverse function $\pays^{-1}$ to output the unique Nash of the proxy game is sufficient for its output to be unique and correct.
\end{proof}


\subsection{Single Item Proportional Weights Social Choice Functions}
\label{s:oneweightspsd}
The goal of this section is to show that every proportional weights
social choice function awarding a single item 
has a price function $\pays$ 
meeting the conditions of \Cref{cor:nash-implementation}.  We state this now as 
the main theoretical result of the paper.

\begin{theorem}
\label{thm:themainresult}
A price function $\pays$ (of equation~\eqref{eq:winner-pays-bid-strats}) -- corresponding to a strictly monotone, continuous, differentiable proportional weights social choice rule -- is an interior $P$-matrix function, and it is uniquely invertible.
\end{theorem}
\begin{proof}
We only need to show that $J_{\pays}$ is an interior $P$-matrix function.  In \Cref{lem:priceid} in \Cref{a:pderi}, we show that under $\pays$, the lower boundaries of the weights space domain are identified boundaries.  \Cref{lem:pderi} (next) shows that when an input is in the lower boundary for any dimension $i$, $J_{\pays}$ has all-zero elements in row $i$, such that its determinant is trivially 0, meeting the (weakened) identified-boundary condition of a $P_0$-matrix.  Otherwise at {\em all} points of the weights space domain, \Cref{thm:mpd2} (at the end of this section) shows that the critical minor of $J_{\pays}$ -- i.e., the minor which removes row/ column indexes corresponding to the dimensions in which its input exists in identified (lower) boundaries -- is strictly a $P$-matrix.
\end{proof}



The rest of this section builds towards \Cref{thm:mpd2}.  We start with the straightforward calculation of the partial derivatives of $\pays$, which in particular give the entries of the Jacobian $J_{\pays}$.  The steps
of the calculations and the proof of \Cref{lem:pderi} are given in~\Cref{a:pderi}.
\ifsoda
\begin{align}
\label{eqn:partialsame}
\frac{\partial \payi}{\partial \wali}(\wals)&=\int^{\wali}_{\wali(0)}\vali'(z)\frac{1}{\wali}\cdot \frac{z}{\left(\sum_k \wali[k]\right)-\wali+z}\\
\notag
&\quad\cdot\left[\frac{\sum_k \wali[k]}{\wali}-1\right]dz\\
\label{eqn:partialcross}
\frac{\partial \payi}{\partial \wali[j]}(\wals)&=\int^{\wali}_{\wali(0)}\vali'(z)\frac{1}{\wali}\cdot \frac{z}{\left(\sum_k \wali[k]\right)-\wali+z}\\
\notag
&\quad\cdot\left[\frac{\sum_k \wali[k]}{\left(\sum_k \wali[k]\right)-\wali+z}-1\right]dz
\end{align}
\else
\begin{eqnarray}
\label{eqn:partialsame}
\frac{\partial \payi}{\partial \wali}(\wals)&=&\int^{\wali}_{\wali(0)}\vali'(z)\frac{1}{\wali}\cdot \frac{z}{\left(\sum_k \wali[k]\right)-\wali+z}\cdot\left[\frac{\sum_k \wali[k]}{\wali}-1\right]dz\\
\label{eqn:partialcross}
\frac{\partial \payi}{\partial \wali[j]}(\wals)&=&\int^{\wali}_{\wali(0)}\vali'(z)\frac{1}{\wali}\cdot \frac{z}{\left(\sum_k \wali[k]\right)-\wali+z}\cdot\left[\frac{\sum_k \wali[k]}{\left(\sum_k \wali[k]\right)-\wali+z}-1\right]dz
\end{eqnarray}
\fi

\begin{lemma}
Given the price function $\pays$ for proportional weights, for $j,k\neq i$, the cross derivatives are the same: $\frac{\partial \payi}{\partial w_j} =\frac{\partial \payi}{\partial w_k}$.  Evaluating the Jacobian at $\wals$, further, all elements of the Jacobian matrix $J_{\pays}$ are positive, i.e., $\frac{\partial \payi}{\partial \wali}>0,~\frac{\partial \payi}{\partial w_j}>0$, except at the $\wali(0)$ lower boundary in dimension $i$ where the elements of row $i$ are $\frac{\partial \payi}{\partial \wali}= \frac{\partial \payi}{\partial w_j}=0$.
\label{lem:pderi}
\end{lemma}


We need to prove that $\pays$ is an interior $P$-matrix function.  Consider weights input $\wals$.  Let $K$ be the count of dimensions $i$ such that coordinate $\wali$ is ``off" the  lower identified boundary in dimension$i$, i.e., $\wali > \wali(0)$.  Without loss of generality we can assume the dimensions of identified boundaries have the largest indexes (if any).

We critically consider only the principal minor of $J_{\pays}$ which results from keeping the first $K$ interior dimensions, as is sufficient to check an interior $P$-matrix function.  We explicitly define the ratio of an agent's ``self-partial" to its ``cross-partial" for any $j\neq i$ by $h_i$, which will be needed for analysis throughout the rest of the paper.\footnote{Technically the $h_i$ terms are functions, each of input $\wals$, but we suppress this in the notation.}
\begin{align}
\label{eq:hi}
h_i&= \frac{\partial \payi}{\partial \wali}/\frac{\partial \payi}{\partial \wali[j]}
\end{align}
\noindent The derivatives that appear are  positive (\Cref{lem:pderi}).  We write the principal minor's Jacobian as\ifsoda~$J_{\pays,K}=D\cdot H =$
\begin{equation}
\label{eq:jacminor}
	\begin{bmatrix}
    \frac{\partial \payi[1]}{\partial \wali[2]} &  0 & \dots  &     0 \\
    0 & \frac{\partial \payi[2]}{\partial \wali[1]} & \dots & 0 \\
    \vdots & \vdots & \ddots & \vdots \\
      0 &  0 & \dots  &  \frac{\partial \payi[K]}{\partial \wali[1]}
\end{bmatrix}
\cdot
\begin{bmatrix}
    h_1 & 1 &  \dots  &     1 \\
    1 & h_2 & \dots & 1\\
    \vdots  & \vdots & \ddots & \vdots \\
      1 &  1 & \dots  &      h_K
\end{bmatrix}
\end{equation}
\else
\begin{equation}
\label{eq:jacminor}
    	J_{\pays,K}=D\cdot H =
	\begin{bmatrix}
    {\partial \payi[1]}/{\partial \wali[2]} &     0 &     0 & \dots  &     0 \\
    0 &     {\partial \payi[2]}/{\partial \wali[1]}  &     0 & \dots  &    0 \\
        0 &     0  &  {\partial \payi[3]}/{\partial \wali[1]} & \dots  &    0 \\
    \vdots & \vdots & \vdots & \ddots & \vdots \\
      0 &        0  &      0  & \dots  &  {\partial \payi[K]}/{\partial \wali[1]}
\end{bmatrix}
\cdot
\begin{bmatrix}
    h_1 &     1 &     1 & \dots  &     1 \\
    1 &    h_2  &     1 & \dots  &    1 \\
        1 &     1  &  h_3 & \dots  &    1 \\
    \vdots & \vdots & \vdots & \ddots & \vdots \\
      1 &        1  &      1  & \dots  &      h_K
\end{bmatrix}
\end{equation}
\fi

Multiplying by a positive diagonal matrix $D$ is a benign operation with respect to the determination of a matrix as a $P$-matrix (\Cref{f:pmatfact}(3)).  We will define $H$ to be the rightmost matrix of equation~\eqref{eq:jacminor} which is composed of $h_i$ elements in the diagonal and all ones elsewhere.  By reduction, we need only show that $H$ is a $P$-matrix, for which it is sufficient to show $H$ is positive definite (\Cref{f:pmatfact}(2)).

We claim the following results, starting with a complete characterization of when an arbitrary matrix $\scrbh$ (with structure of $H$) is positive definite, a result which could be of independent interest.

\begin{theorem}
\label{thm:mpd}
Consider a $K\times K$ matrix $\scrbh$ with diagonal $\scrlh_1, \scrlh_2, ..., \scrlh_K$ and all other entries equal to 1 (and without loss of generality $\scrlh_1 \leq \scrlh_2 \leq \ldots \leq \scrlh_K$).  
The following is a complete characterization describing when $\scrbh$ is
positive definite.
\begin{enumerate}
\item if $\scrlh_1 \leq 0$, then the matrix $\scrbh$ is not positive definite;
\item if $\scrlh_1 \geq 1$ and $\scrlh_2 > 1$, then $\scrbh$ is positive definite;
\item if $0<\scrlh_1,\scrlh_2 \leq 1$, then $\scrbh$ is not positive definite;
\item if $0<\scrlh_1< 1$ and $\scrlh_2>1$, then $\scrbh$ is positive definite if and only if $\sum_k \frac{1}{1-\scrlh_k} > 1$.
\end{enumerate}
\end{theorem}

The proof of \Cref{thm:mpd} is given in
Appendix~\ref{as:hmatrixpsd} where the main difficulty is part~(4).
\Cref{thm:mpd} is for arbitrary $\scrbh$.  We now return to the
specific consideration of $H$ resulting from $\pays$ and
$J_{\pays,K}$, showing in \Cref{thm:mpd2} that it must be covered by cases (2) or (4) from
\Cref{thm:mpd}.  The proofs of \Cref{lem:halfw} and
\Cref{lem:fin} are given in \Cref{as:thmmpd}.

\begin{lemma}
\label{lem:halfw}
If $h_i\leq 1$, then $w_i> 0.5\sum_k w_k$, and all other weights must have $w_j < 0.5\sum_k w_k$ (and $h_j>1$).
\end{lemma}

\begin{lemma} When $h_1<1$ and $h_j>1~\forall j\neq1$, we have $\sum_k \frac{1}{1-h_k}>1$.
\label{lem:fin}
\end{lemma}

\begin{theorem}
\label{thm:mpd2}
Let matrix $J_{\pays}$ be the Jacobian of $\pays$ at weights $\wals$
of a positive, strictly increasing, and differentiable proportional
weights social choice functions.  $\pays$ is an interior $P$-matrix function.
\end{theorem}
\begin{proof}
By the definition of an interior $P$-matrix function
(\Cref{def:interiorpd}), we consider the restriction to the minor
$J_{\pays,K}=D\cdot H$ at $(\wali[1],\ldots,\wali[K])$, where coordinates in identified (lower) boundaries of weights
space have been discarded (see equation~\eqref{eq:jacminor}).
Because weights $(\wali[1],\ldots,\wali[K])$ are definitively off their respective lower boundaries,
\Cref{lem:pderi} implies that all $h_i\in\left\{h_1,\ldots,h_K\right\}$
are strictly positive.  By \Cref{lem:halfw}, at most one agent
$i$ has $h_i\leq 1$.  Without loss of generality, we can set this
$i=1$.  So there are just two cases:
\begin{enumerate}
\item $h_1\geq1$ and $h_j>1~\forall j\neq 1$, and
\item $0<h_1<1$ and $h_j>1~\forall j\neq 1$.
\end{enumerate}
These are respectively cases (2) and (4) of \Cref{thm:mpd}.  To
satisfy the condition within case (4) of \Cref{thm:mpd},
\Cref{lem:fin} is sufficient.  Thus, the factor $H$ of the
Jacobian minor $J_{\pays,K}$ is positive definite.  Finally, using \Cref{f:pmatfact}, $H$ is a $P$-matrix and the
product $J_{\pays,K}=D\cdot H$ is also a $P$-matrix.
\end{proof}

\subsection{Impossibility Results for Complementarities}
\label{s:hardness}

In this section we show that for a natural generalization of the proportional weights social choice function to an environment with complementarities between agents, the values of the agents cannot necessarily 
be identified from the prices output by the mechanism.


The impossibility result we present will consider a generalization of exponential weights to environments with complementarities.  We will consider the special case where the agents are partitioned and the mechanism can allocate to all agents in any one part, but agents from multiple parts may not be simultaneously allocated. We prove that a natural extension of exponential weights to partition set systems results in a price function $\pays$ that is not one-to-one, by counterexample.  Thus, the price function is generally not invertible: no algorithm can distinguish between two (or more) valuation profiles which give the same prices.

\begin{definition}
The exponential weights social choice function for an $n$-agent
partition set system with parts ${\cal S} = (S_1,\ldots, S_r)$ is given by:
\begin{itemize}
\item $\vali[S] = \sum_{i \in S} \vali$ for $S\in\cal{S}$;
\item $\alloci[S](\vals) = \frac{e^{v_S}}{\sum_{T \in {\cal S}} e^{v_T}}$ for $S\in \cal{S}$;
\item $\alloci(\vals) = \alloci[S](\vals)$ for $i \in S$
\end{itemize}
\end{definition}

The resulting price function corresponding to the exponential weights social choice function for partition set systems is
\ifsoda
\begin{align*}
\payi(\vals)&=\vali-\frac{\int^{\vali}_0 \alloci(z,\valsmi)dz}{\alloci(\vals)}\\
&=\vali-\frac{\sum_{T} e^{v_T}}{e^{v_S}}\int^{\vali}_0 \frac{e^ze^{v_{S \setminus \{i\}}}}{e^ze^{v_{S \setminus\{i\}}}+\sum_{T\neq S} e^{v_T}}dz\\
&=\vali-\frac{\sum_{T} e^{v_T}}{e^{v_S}}\left[\ln \left(\sum\nolimits_{T} e^{v_T}\right) \right.\\
&\qquad\qquad\qquad\qquad \left.- \ln\left(e^{v_{S \setminus \{i\}}}+\sum\nolimits_{T\neq S} e^{v_T}\right)\right]
\end{align*}
\else
\begin{align*}
\payi(\vals)&=\vali-\frac{\int^{\vali}_0 \alloci(z,\valsmi)dz}{\alloci(\vals)}\\
&=\vali-\frac{\sum_{T} e^{v_T}}{e^{v_S}}\int^{\vali}_0 \frac{e^ze^{v_{S \setminus \{i\}}}}{e^ze^{v_{S \setminus\{i\}}}+\sum_{T\neq S} e^{v_T}}dz\\
&=\vali-\frac{\sum_{T} e^{v_T}}{e^{v_S}}\left(\ln \left(\sum\nolimits_{T} e^{v_T}\right) - \ln\left(e^{v_{S \setminus \{i\}}}+\sum\nolimits_{T\neq S} e^{v_T}\right)\right)
\end{align*}
\fi

\noindent The completion of the counterexample is in the following lemma.

\begin{lemma}
\label{lem:thmupperbound}
The price function $\pays$ corresponding to the exponential weights social choice function for partition set systems (with at least one partition containing two or more agents) is not one-to-one.
\end{lemma}

\begin{figure}[t]
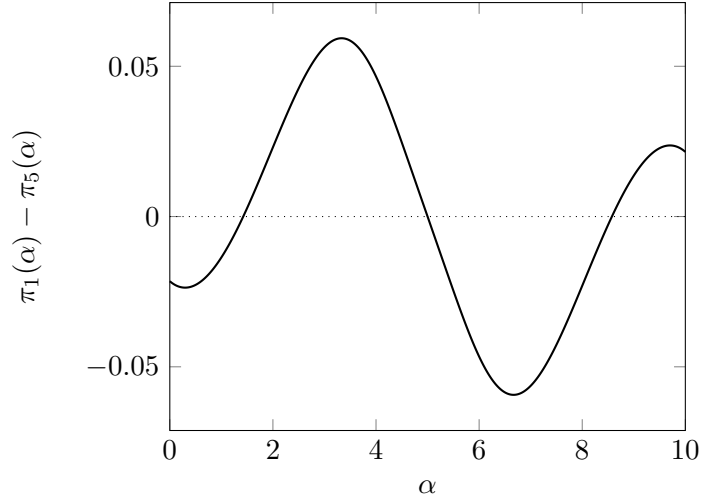

  \centering
   \multipleinversesfigure
  \caption{Graphing the function $[\payi[1](\alpha)-\payi[5](\alpha)]$
    from the proof of \Cref{lem:thmupperbound}.  The zeroes of the
    function parameterize values for agents in $S_1$ and $S_2$ such that all agents across both parts have identical prices, despite the agents of each group having strictly distinct values from each other.  (Note, by design, the curve is rotationally symmetric around the point $(5,0)$.)}
  \label{fig:fpa}
\end{figure}

\begin{proof}

We prove that the price function is not one-to-one (and consequentially by the contrapositive of~\Cref{cor:nash-implementation} its Jacobian is not positive definite).  We first set up a parameterized analysis and then choose the parameters later.

Let there be $k$ agents in set $S_1$ who all have the same valuation $\alpha/k$, and another $k$ agents in set $S_2$ who all have the same valuation $(\beta-\alpha)/k$. Note $\beta=
v_{S_1}+v_{S_2}$. Players in all other sets $S_r$ for $r>2$ have a constant value $\val_{\text{others}}$ and can be summarized by a single parameter $\delta$ by letting $e^\delta=\sum_{r>2} e^{v_{S_r}}$.  Parameters $k,~\alpha,~\beta$
and $\val_{others}$ will be determined later.

The price for agent 1 in part $S_1$ is
\ifsoda
\begin{align*}\payi[1]&=\frac{\alpha}{k}-\frac{e^{\alpha}+e^{\beta-\alpha}+e^\delta}{e^{\alpha}}\Big[\ln (e^{\alpha}+e^{\beta-\alpha}+e^\delta)\\
&\qquad\qquad\qquad\qquad - \ln(e^{(1-1/k)\alpha}+e^{\beta-\alpha}+e^\delta)\Big]
\end{align*}
\else
$$\payi[1]=\frac{\alpha}{k}-\frac{e^{\alpha}+e^{\beta-\alpha}+e^\delta}{e^{\alpha}}[\ln (e^{\alpha}+e^{\beta-\alpha}+e^\delta) - \ln(e^{(1-1/k)\alpha}+e^{\beta-\alpha}+e^\delta)]$$
\fi
The price for agent $k+1$ in part $S_2$ is
\ifsoda
\begin{align*}\payi[k+1]&=\frac{\beta-\alpha}{k}-\frac{e^{\alpha}+e^{\beta-\alpha}+e^\delta}{e^{\beta-\alpha}}\Big[\ln (e^{\alpha}+e^{\beta-\alpha}+e^\delta)\\
&\qquad\qquad\qquad\qquad - \ln(e^{(1-1/k)(\beta-\alpha)}+e^{\alpha}+e^\delta)\Big]
\end{align*}
\else
$$\payi[k+1]=\frac{\beta-\alpha}{k}-\frac{e^{\alpha}+e^{\beta-\alpha}+e^\delta}{e^{\beta-\alpha}}[\ln (e^{\alpha}+e^{\beta-\alpha}+e^\delta) - \ln(e^{(1-1/k)(\beta-\alpha)}+e^{\alpha}+e^\delta)]$$
\fi
We now show that it is possible that player $1$ and player $k+1$ have different valuations but are charged the same prices.  Consider the case $k=4,~\beta=10$, and there is one additional part $S_3$ with a single agent 9 with $\vali[9] = \val_{\text{others}}=4$ inducing $\delta=4$.  Then we can consider the quantity $(\payi[1]-\payi[5])$ as a function of parameter $\alpha$ with $\vali[1] = \ldots =\vali[4] = \alpha/4$ and $\vali[5] = \ldots = \vali[8] = (10-\alpha)/4$.

This function $[\payi[1])\alpha)-\payi[5](\alpha)]$ is graphed in Figure~\ref{fig:fpa}, where we can see that there are three solutions for $\payi[1]=\payi[5]$.
Without showing the explicit calculation, $\payi[1]=\payi[5]$ holds for a value profile where $\vali[1] = \ldots = \vali[4]=\alpha/4\approx 0.375$ and 
$\vali[5] = \ldots = \vali[8] = (\beta-\alpha)/4\approx 2.125$,
and $v_9=4$. In this case, the seller cannot distinguish between $S_1$ and $S_2$ which part has agents with identical values $\approx 0.375$ versus the other part whose agents all have values $\approx 2.125$.
\end{proof}

This lemma can be generalized as follows.  A set system is
downward-closed if all subsets of feasible sets are feasible.  Agents
are substitutes if the set system satisfies the matroid augmentation
property, i.e., for any pair of feasible sets with distinct
cardinalities, there exists an element from the larger set that is not in the
smaller set that can be added to the smaller set and the resulting set
remains feasible.  A set system exhibits complementarities if agents
are substitutes (i.e., there exist sets that fail the augmentation
property).  Exponential weights can be generalized to any set system
by choosing a maximal set with probability proportional to its
exponentiated weight.  The impossibility result above can then be
easily generalized to any set system that exhibits complementarities
by identifying the sets and taking $S_1$ and $S_2$ to be the agents
uniquely in each set (i.e., not in their intersection), and setting
all other agent values to zero.

\section{Computational Methods for Inverting the Price Function}
\label{s:computation}

In \Cref{s:nasheq} we gave the price-inversion
algorithm (\Cref{d:nash-mech}), which is a well-defined, continuous function that inverts the payment identity $\pays$ to map prices $\rezs$ back to values $\vals$ (\Cref{thm:themainresult}).   The algorithm is straightforward except for Step~\ref{step:nash} which requires the computation of a Nash equilibrium in the defined proxy game. In this section we give a simple algorithm for identifying an equilibrium of the proxy game and thus show that the inverse function can be efficiently computed.

The algorithm for solving the proxy game is enabled by two
observations.  First, for player $i$, the sum of weights $s=\sum_k
\wali[k]$ summarizes everything that needs to be known about the other
players and this observation leads to a many-to-one reduction in the
dimension of search space.  Consequently, the price function can be
rewritten as a function $\bpayi(s,\wali)$.\footnote{See
  equation~\eqref{eqn:stratiofswi} in~\Cref{a:supportingalg} for its formal definition.}
Second, because the price function $\pays$ is invertible, the sum $s$
is uniquely determined from the prices.

Obviously at most one agent can have strictly more than half the total
weight $s$.  For the rest of this section, without loss of generality
we fix agent $i^*$ to mean that $\wali[i^*]$ is not restricted and
$\wali \leq s/2$ for all $i\neq i^*$.

Fix observed input prices $\rezs$.  For any agent $i\neq i^*$, consider the set of points $(s,\wali)$ for which $\bpayi$ outputs $\rezi$.  
Our first key Lemma~\ref{lem:decreasingweightj} (below) will show that, restricting to the space $\wali \leq s/2$, this set of points can be interpreted as a real-valued, monotone decreasing function of $s$, denoted $\zwali(\cdot)$.
With this property holding for all agents other than $i^*$, we can
express the price function for agent $i^*$ with dependence on
prices $\rezsmi[i^*]$ and sum $s$:\footnote{Regarding functions $\wali^{\rezs}$ and $\zpayi[i^*]$.  We write them both parameterized by vector $\rezs$ to demark them with a simple notation, in a common way because their usage is always related.  Note however, $\rezs$ implies an over-dependence on parameters.  $\wali^{\rezs}$ only uses $\rezi$, and $\zpayi[i^*]$ uses all of $\rezsmi[i^*]$ but not $\rezi[i^*]$.}

\begin{equation}
\label{eqn:estbid}
\zpayi[i^*](s) \vcentcolon=\payi[i^*](\max\{s-\sum\nolimits_{i\neq i^*} \zwali(s),\wali[i^*](0) \}, \zwals_{-i^*}(s))
\end{equation}
where, for guess of total weight $s$, the quantity $s-\sum_{i\neq i^*} \zwali(s)$
assigns an intermediate guess of $\wali[i^*]$ as the ``balance'' of
the quantity $s$ having subtracted the implied weights of the
``small'' agents for guess $s$.\footnote{When the guess $\wali[i^*]=s-\sum\nolimits_{i\neq i^*} \zwali(s)$ is irrationally small or even negative, the structure of the problem allows us to round it up to constant $\wali[i^*](0)$, sufficiently preserving monotonicity.  See the proof of~\Cref{lem:increasingstratbar} in~\Cref{a:b:proofs}.}  Our second key
\Cref{lem:increasingstratbar} (below) will show that, on the range of $s$ for
which it is well-defined, the function $\zpayi[i^*]$ is strictly
monotonically increasing.

This setup suggests a natural binary search procedure.  For some agent
$i^*$ and small initial guess of $s$, the implied price for $i^*$ is smaller than the observed input, i.e. $\zpayi[i^*](s)<\rezi[i^*]$.  A large guess of $s$ implies too big of a price for $i^*$, and monotonicity will then guarantee a crossing.  The algorithm has the
following steps:

\begin{enumerate}
\item Find an agent $i^*$ by iteratively running the following for each fixed assignment of agent $i \in \left\{1,\ldots,n\right\}$:
\begin{enumerate}
\item temporarily set $i^*=i$;
\item determine the range of $s$ on which $\zpayi[i^*]$ is well-defined and searching is appropriate;\footnote{In the proper algorithm and proof, we will give better bounds on the range of the search; for now, as a simple indication that bounds exist, note that there exists a solution for some appropriate $i^*$ within the general bounds on $s$ as $\sum_k \wali[k](\rezi[k])\leq \sum_k \wali[k](\vali[k]) =s\leq \sum_k \wali[k](h)$ for known $\rezi$ and max value $\maxval$, because $\rezi\leq \vali \leq h$.}
\item if this range of $s$ is non-empty, permanently fix $i^* = i$ and break the for-loop;
\end{enumerate}
\item use the monotonicity of $\zpayi[i^*]$ to binary search on $s$ for the true $s^*$, converging $\zpayi[i^*](s)$ to $\rezi[i^*]$;
\item when the binary search has been run to satisfactory precision and reached a final estimate $\tilde{s}$,
output weights $\tilde{\wals}=(\tilde{s}-\sum_{i\neq i^*} \zwali(\tilde{s}),\zwalsmi[i^*](\tilde{s}))$ which invert to values $\evals$ via respective $\vali(\cdot)$ functions.
\end{enumerate}

\noindent The rest of this section formalizes our key results.

\subsection{Computation through Total Sum Weights}
\label{s:compsum}


The following theorem claims correctness of the algorithm, and is the
main result of this section.

\begin{theorem}
\label{thm:simplealg}
Given weights $\wals$ and payments $\rezs = \pays(\wals)$ according to
a proportional weights social choice function, the algorithm
identifies weights $\tilde{\wals}$ to within $\epsilon$ of the true
weights $\wals$ in time polynomial in the number of agents $n$, the
logarithm of the ratio of high to low weights $\max\nolimits_i\ln
(\wali(h)/\wali(0))$, and the logarithm of the desired precision $\ln
1/\epsilon$.  
\end{theorem}

A major object of interest for this sequence of results is the price level set defined by $\mathcal{Q}^{\rezs}_i = \left\{(s,\wali)\,|\,\bpayi(s,\wali)=\rezi\right\}$, i.e., all of the $(s,\wali)$ pairs mapping to the price $\rezi$ under $\bpayi$, and also in particular its subset $\mathcal{P}^{\rezs}_i=\left\{(s,\wali)\,|\,\bpayi(s,\wali)=\rezi ~\text{and}~ \wali\leq s/2\right\}\subseteq\mathcal{Q}^{\rezs}_i$ which restricts the set to the region where $\wali$ is at most half the total weight $s$.  Define $\minsumi = \min \{s:(s,\wali) \in \mathcal{P}^{\rezs}_i\}$ as the lower bound on the sum $s$ on which the set $\mathcal{P}^{\rezs}_i$ is supported.  These quantities are depicted in \Cref{fig:pricelevel}.

 \begin{figure*}
  \centering
\pricelevelfigure
\caption{The price level set curve $\mathcal{Q}^{\rezs}_i =
  \left\{(s,\wali):\bpayi(s,\wali)=\rezi\right\}$ (thick, gray,
  dashed), is decreasing below the $\wali = s/2$ line
  (\Cref{lem:decreasingweightj}) where it is defined by its subset
  $\mathcal{P}^{\rezs}_i$ (thin, black, solid).  It is bounded above
  by the $\wali = s$ line (trivially as $s$ sums over all weights) and
  the $\wali = \wali(h)$ line (the maximum weight in the support of
  the values), and below by the $\wali = \wali(0)$ line which we have
  assumed to be strictly positive.  $\minsumi$ is the minimum weight-sum consistent
  with observed price $\rezi$ and weights $\wali \leq s/2$.}
  \label{fig:pricelevel}
\end{figure*}

We give the formal statements of the two most critical lemmas
supporting \Cref{thm:simplealg}.

\begin{lemma}
\label{lem:decreasingweightj}
The price level set $\mathcal{Q}^{\rezs}_i$ is a curve; further,
restricting $\mathcal{Q}^{\rezs}_i$ to the region $\wali \leq s/2$,
the resulting subset $\mathcal{P}^{\rezs}_i$ can be written as
$\{(s,\zwali(s)):s \in [\minsumi,\infty)\}$ for a real-valued decreasing function $\zwali$ mapping sum $s$ to a weight $\wali$ that is parameterized by the observed price $\rezi$.
\end{lemma}

\begin{lemma}
\label{lem:increasingstratbar}
For any agent $i^*$ and $s\in [\max_{j\neq i^*}\minsumi[j],\infty)$, function $\zpayi[i^*]$ is weakly increasing; specifically, $\zpayi[i^*]$ is constant when $s-\sum\nolimits_{i\neq i^*} \zwali(s)\leq\wali[i^*](0)$ and strictly increasing otherwise.
\end{lemma}

A key step in the proof of \Cref{lem:increasingstratbar} will depend
on \Cref{lem:fin}.  The $\frac{1}{1-h_k}$ terms
in the statement of \Cref{lem:fin} are realized to correspond to derivatives
of $\zwali[k]$ functions.  Consequently, the correctness of the
algorithm critically relies on the proof of a unique inverse to the price function.  

We give the proofs of \Cref{thm:simplealg},
\Cref{lem:decreasingweightj}, and \Cref{lem:increasingstratbar} in
\Cref{a:b:proofs}.  Preceding these proofs is supporting material: \Cref{a:b:structure} gives a detailed analysis of the
structure of the search space, and \Cref{a:b:alg} gives the description of the binary search algorithm with full details.

\bibliographystyle{apalike}
\bibliography{bib}

\begin{appendix}

\section{Supporting Material for \Cref{s:nasheq}}

\subsection{Proofs of \Cref{thm:gnext} and \Cref{thm:gameunique}}
\label{a:gnext}

Before getting to results, we define a {\em dimensionally-reduced function} by a parameterized procedure.  This procedure will be useful as a sub-routine in multiple proofs.

\begin{definition}
\label{def:dimreduct}
Given a function $f:\Omega\to\reals^n$, two points $\bomega^1,\bomega^2$ in compact and convex product space $\Omega\subset\reals^n$, and a set $K$ of dimensions with identified lower boundaries with cardinality $k=|K|$.  Define a {\em dimensionally-reduced function} $D:\Omega_{-K}\to\reals^{n-k}$ where
\begin{itemize}
    \item $\Omega_{-K}\subset\Omega$ is the projection of product space $\Omega$ into dimensions not in $K$, and further the lower bounds of each remaining dimension $i$ is (weakly) increased to $\min\{\omega^1_i,\omega^2_i\}$ respectively, and analogously each upper bound decreased to $\max\{\omega^1_i,\omega^2_i\}$; 
    \item $D(\bomega_{-K})=f(\bomega_{-K},\bomega_K=\mathbf{c}_K)$ for (vector) $\mathbf{c}_K$ the fixed inputs of (removed dimension) identified boundaries, input to $f$ as constant parameters.
\end{itemize}
\end{definition}

\noindent 
We restate and prove \Cref{thm:gnext} here.  Recall it is an extension of \Cref{thm:psdunique} \citep{GN-65}.  
Its proof depends on \Cref{lem:identify} given immediately following.

\begin{numberedtheorem}{\ref{thm:gnext}}
If function $f:\Omega\to\reals^n$ on compact and convex product domain $\Omega\subset\reals^n$ is an interior $P$-matrix function (Definition~\ref{def:interiorpd}), then it is one-to-one, and therefore invertible on its image.
\end{numberedtheorem}
\begin{proof}
By contradiction, assume there exist two distinct inputs $\bomega^1,~\bomega^2$ such that $f(\bomega^1)=f(\bomega^2)$.  With equal outputs under $f$ it must be that 
$\bomega^1,~\bomega^2$ exist in the same set of identified boundaries because 
by \Cref{lem:identify} (given next), 
an input in each of these identified boundaries outputs a unique constant in its respective dimension.  Let the common set of dimensions in identified boundaries be $K$.  We consider dimensionally-reduced function $D$ applied to $f, ~\bomega^1,~\bomega^2$ and set $K$ (of \Cref{def:dimreduct}).  $D$ now meets all of the conditions of \Cref{thm:psdunique} (in particular $D$ has Jacobian as a $P$-matrix everywhere because no coordinate of $\bomega^1_{-K}$ or $\bomega^2_{-K}$ is in the original identified boundaries, and $f$ is an interior $P$-matrix function which must be a strict $P$-matrix function when excluding identified boundaries by definition).  Therefore $D$ is one-to-one on its restricted domain, which includes $\bomega^1_{-K},~\bomega^2_{-K}$, a result which extends to analysis under $f$ such that $f$ must also be one-to-one.  I.e., it must be that $f(\bomega^1)\neq f(\bomega^2)$ in some coordinate outside the set $K$, giving the contradiction.
\end{proof}

\begin{lemma}
\label{lem:identify}
If a function $f:\Omega\to\reals^n$ on domain $\Omega\subset\reals^n$ is an interior $P$-matrix function (of Definition~\ref{def:interiorpd}), then function $f_i$ evaluates to constant $f_i(c_i,\cdot)$ on an identified boundary in dimension $i$ with coordinate $c_i$ if and only if the input $\bomega$ to $f_i$ has $\omega_i = c_i$.
\end{lemma}
\begin{proof}
Without loss of generality, assume $c_i = a_i$ the lower boundary in dimension $i$, with the upper boundary argument by symmetry.  For sufficiency, note that by definition of an identified boundary (\Cref{def:identbound}), all cross-partials on the function $f_i$ (evaluated at the identified boundary) are identically 0.  
I.e., $\partial f_i/\partial \omega_j(a_i,\bomegasmi) = 0$ for all $j\neq i$ and for all $\bomegasmi$.  Therefore $f_i(a_i,\bomegasmi)$ is a constant.

For necessity, consider an input $(d_i,\bomegasmi)$ with $d_i>a_i$ ``off" the identified boundary.  Evaluated at all inputs $\omega_i\in\left(a_i,d_i\right]$, the self-partial $\partial f_i/\partial \omega_i>0$ is necessary by the assumption that $f$ is an internal $P$-matrix function, because an implication of its definition is that, when $\omega_i$ is not in the lower boundary, the diagonal element of the Jacobian $\partial f_i/\partial \omega_i$ at index $(i,i)$ must be strictly positive (because diagonal elements are principal minors with dimension $1\times 1)$.  

Therefore $f_i(d_i,\bomega_{-i})=f_i(a_i,\bomega_{-i})+\int_{a_i}^{d_i}\left[\partial f_i /\partial \omega_i(z,~\bomegasmi) \right]dz>f_i(a_i,\bomega_{-i})$ because the integral of a strictly positive function is strictly positive when $d_i>a_i$.
\end{proof}

\noindent Note, an implication of \Cref{lem:identify} in the context of our price functions is that we observe agent $i$ to have price $\rezi=0$ if and only if agent $i$ had minimal weight $\wali(0)$, trivially implying value $\vali=0$.

The rest of this appendix section is devoted to proving \Cref{thm:gameunique}.  Additionally, we develop the following corollary, which should be of independent interest to the game theory community.

\begin{corollary}
\label{thm:nmunique}
A game with a compact and convex product action space and pseudogradient that is an $N$-matrix function has a unique Nash equilibrium, which is pure.
\end{corollary}

\noindent A significance of \Cref{thm:nmunique} is that it extends a classic result by \citet{ros-65}.

\begin{theorem}[\citealp{ros-65}]
\label{thm:rosenclassic}
A game with a compact and convex product action space $\Omega$ and pseudogradient $\left[\frac{\partial U_i}{\partial \omega_i}\right]_i$such that for all inputs $\bomega^1,~\bomega^2\in\Omega$:
\begin{equation*}
    \left(\left[\frac{\partial U_i}{\partial \omega_i}\right]_i(\bomega^2)-\left[\frac{\partial U_i}{\partial \omega_i}\right]_i(\bomega^1)\right)\cdot (\bomega^2-\bomega^1) < 0
\end{equation*}
has a unique Nash equilibrium, which is pure.
\end{theorem}

We continue by listing three results from \citet{GN-65}.  The first, \Cref{thm:gnpartdiagconc}, is a result which appears in their paper.   The third restates their result which we have already given as \Cref{thm:psdunique} in this paper.  The second, \Cref{thm:gndiagconc}, is a new intermediate sub-result statement, which summarizes the preliminary analysis within \citeauthor{GN-65}'s proof of \Cref{thm:psdunique}.  \Cref{thm:gndiagconc} is a generalization of \Cref{thm:gnpartdiagconc}.\footnote{An organizational note on numbering of theorems: our \Cref{thm:psdunique} is given as Theorem 4 in the \citeauthor{GN-65} paper; our \Cref{thm:gnpartdiagconc} is their Theorem 3.  Our \Cref{thm:gndiagconc} is their result but is not an explicit statement.}

\Cref{thm:gndiagconc} is indispensable for our \Cref{thm:gameunique} and \Cref{thm:nmunique} results, yet a proof for this statement explicitly does not exist in continuous, cohesive form.  To spare the reader the task of personally piecing it together, we give the proof here, adapted from \citeauthor{GN-65}.  For completeness, we will then finish the proof of \Cref{thm:psdunique} which basically becomes a corollary.

\begin{theorem}[\citealp{GN-65}]
\label{thm:gnpartdiagconc}
If function $f:\Omega\to\reals^n$ on compact and convex product domain $\Omega\subset\reals^n$ has Jacobian $J_f$ which is a $P$-matrix at every $\bomega\in\Omega$, then for any fixed input  $\bomega^1\in\Omega$, and variable $\bomega^2$ from the domain $\Omega$, the inequalities
\begin{equation*}
    f(\bomega^1)\leq f(\bomega^2),~ \bomega^1\geq \bomega^2
\end{equation*}
have only the solution $\bomega^1=\bomega^2$.
\end{theorem}

\noindent \Cref{thm:gnpartdiagconc} has an interpretation in the context of our proxy games, with $\bomega$ as a vector of actions, and $f$ as the pseudogradient function on utilities.  For games maximizing utility we would use the equivalent analogous statement for Jacobian as $N$-matrix everywhere (and flip the sign of the first vector inequality).  What the N-matrix version of \Cref{thm:gnpartdiagconc} says when it holds for a game is: given $\bomega^2$, there can not exist distinct pointwise ``weakly larger" actions $\bomega^1$ such that all local preference gradients (with respect to own action) are also weakly larger at $\bomega^1$ compared to $\bomega^2$.

However there is nothing special about the ``weakly larger" direction -- i.e., the ``all-positives" orthant.  The pure-math interpretation of \Cref{thm:gnpartdiagconc} (still for $N$-matrix) is that ``moves" from $\bomega^2$ in the direction of the all-positive orthant to $\bomega^1$ can not also move the output in the direction of the all-positive orthant. 
The generalization says, given an $N$-matrix Jacobian everywhere, {\em moving the input in the direction of any orthant can not also move the output in the direction of the same orthant, i.e., by the contrapositive, there must exist a dimension in which the change in the input and the change of the corresponding output have opposite signs}.  This idea is immediately pertinent in game theory with actions as inputs and utility gradients as outputs, as the basis of a technique to contradict two action profiles supposedly both being in equilibrium.

We state this intermediate result formally here with \Cref{thm:gndiagconc} (but for continuity of language in result statements, we write it as the $P$-matrix version).  To repeat, the proof here mirrors the first steps of \citeauthor{GN-65}'s proof of \Cref{thm:psdunique}, with slight re-working to be explicitly restated as a generalization of \Cref{thm:gnpartdiagconc}.  Note the following definition for use in \Cref{thm:gndiagconc}.

\begin{definition}
\label{def:binop}
Define the operators $\mathds{1},\mathds{-1}$ applied to inequalities by: multiplying an inequality by $\mathds{1}$ leaves it unchanged, and multiplying it by $\mathds{-1}$ reverses the sign of the inequality.
\end{definition}

\begin{theorem}[\citealp{GN-65}]
\label{thm:gndiagconc}
If function $f:\Omega\to\reals^n$ on compact and convex product domain $\Omega\subset\reals^n$ has Jacobian $J_f$ which is a $P$-matrix at every $\bomega\in\Omega$,  then for any fixed input  $\bomega^1\in\Omega$, and variable $\bomega^2$ from the domain $\Omega$, for every binary vector $\mathds{B}\in\left\{\mathds{1},\mathds{-1}\right\}^n$ the inequalities
\begin{align*}
    \mathds{B}_1\left(f_1(\bomega^1)\right.&\left.\leq f_1(\bomega^2)\right),~\mathds{B}_1\left(\omega^1_1\geq \omega_1^2\right)\\
    & \vdots&\\
    \mathds{B}_n\left(f_n(\bomega^1)\right.&\left.\leq f_n(\bomega^2)\right),~\mathds{B}_n\left(\omega^1_n\geq \omega^2_n\right)
\end{align*}
have only the solution $\bomega^1=\bomega^2$.  Equivalently (the contrapostive), given inputs $\bomega^1,\bomega^2\neq \bomega^1$, there must exist a dimension $i$ such that $(\omega^1_i-\omega^2_i)\cdot(f_i(\bomega^1)-f_i(\bomega^2))>0$.
\end{theorem}
\begin{proof}
Note that we will write the proof to parallel the argument as given by \citeauthor{GN-65}, and connect it back to the binary vector $\mathds{B}$ as appropriate.

Inputs $\bomega^1,~\bomega^2\in\Omega$ are explicitly indexed by $\left(\omega^1_1,\ldots,\omega^1_n\right)$ and $\left(\omega^2_1,\ldots,\omega^2_n\right)$.  By contradiction, assume $\bomega^1,~\bomega^2$ are distinct but there exists vector $\mathds{B}^*$ such that all of the inequalities listed in the theorem statement are satisfied.

Without loss of generality we may assume there exists index $k$ such that $\omega^2_i\leq \omega^1_i$ for $i\leq k$ and $\omega^2_i\geq\omega^1_i$ for $i>k$.  If $k=n$ (or by symmetry $k=0$) then we are in the exact setting of \Cref{thm:gnpartdiagconc} (here with the vector $\mathds{B}^*=\left\{\mathds{1}\right\}^n$), which requires $\bomega^1=\bomega^2$.

So from here on we assume $0<k<n$.  To satisfy the second inequality in each line of the set of inequalities in the theorem statement, it must be that $\mathds{B}^*$ is the vector of $k$ $\mathds{1}$s followed by $(n-k)$ $\mathds{-1}$s.  Define the analogous mapping $D:\reals^n\to\reals^n$ by
\begin{equation*}
    D(\omega_1,\ldots,\omega_n)=(\omega_1,\ldots,\omega_k,-\omega_{k+1},\ldots,-\omega_n)
\end{equation*}
\noindent Clearly $D$ is a bijection on $\reals^n$ with inverse $D^{-1}=D$, and further $D(\Omega)$ is still a compact and convex product space.  Let $E:D(\Omega)\to\reals^n$ be the composite mapping $E=D\circ f\circ D$.  (I.e., the function $E$ on the domain $D(\Omega)$ operates as follows: the first application of $D$ maps back to $\Omega$, to which $E$ can then properly apply $f$, and finally $D$ is applied again to this output.)  At this point, we confirm that the following inequalities hold by inspection, because the application of $\mathds{B}^*$ to the system of inequalities in the theorem statement dovetails with the use of the mapping $D$.
\begin{equation}
    \label{eqn:orthantmap}
E(D(\bomega^1))\leq E(D(\bomega^2)),~D(\bomega^1)\geq D(\bomega^2)
\end{equation}

\noindent The Jacobian $J_E$ of $E$ is a $P$-matrix because it is obtained from the Jacobian $J_f$ by simple changes of row/column signs which preserve the classification as $P$-matrix.  We use \Cref{lem:signflipping} to make this explicit (given immediately following this proof).  In comparison to the Jacobian of $f$, the Jacobian of $E$ is obtained by multiplying each row and each column of $f$ with index at least $k+1$ by a factor of $-1$.  If we ``transform" the Jacobian of $f$ into the Jacobian of $E$ by considering each $i> k$ in turn one step at a time, by multiplying the $i$ row and $i$ column each by $-1$ in each one step, we have that the resulting matrix is still a $P$-matrix as an invariant after each step (by \Cref{lem:signflipping}), such that $J_E$ is a $P$-matrix when the transformation concludes.

With $J_E$ a $P$-matrix and equation~\eqref{eqn:orthantmap}, we can invoke \Cref{thm:gnpartdiagconc} to conclude that $D(\bomega^1)=D(\bomega^2)$, which immediately implies that $\bomega^1=\bomega^2$ by applying $D^{-1}$ to both sides.  This gives the desired contradiction, as the analyzed contradiction also holds by analogy for $f$.
\end{proof}

\begin{lemma}
\label{lem:signflipping}
Given $K\times K$ matrix $M$ as a $P$-matrix.  For any index $i\in\left\{1,...K\right\}$, the matrix $M'$ resulting from multiplying row $i$ by $-1$ and successively column $i$ by $-1$ is also a $P$-matrix.
\end{lemma}
\begin{proof}
As a first note, the element of matrix $M'$ at index $(i,i)$ gets multiplied by $-1$ in both the row-multiplication and column-multiplication operations, so its sign remains unchanged.  All other elements of either the $i$ row or $i$ column have sign flipped from $M$.

Consider within matrix $M'$, the determinant of {\em any} principal minor $M''$ of $M'$, including possibly $M'$ itself.  Without loss of generality, the following argument holds for any $M''$, we don't need to explicitly consider any missing indexes from the original $M'$.  First in particular, if $M''$ excludes row/column $i$ then its determinant remains unchanged.

Otherwise we use the algebraic definition of a determinant.  The determinant of $M''$ is a sum over product-terms with the following property: each product-term includes exactly one element from each row and each column of $M''$, and includes such exhaustively.  Any such additive product-term (within the sum making up the determinant calculation) that includes the element of $M'$ at index $(i,i)$ can not include any other element of $M'$ from row $i$ or column $i$, therefore this term is exactly equal to the respective principal minor determinant term when calculated for the matrix $M$.

Any additive term that does not include the element of $M'$ at index $(i,i)$ must use some term as $(i,x)$ and also some term as $(y,i)$ for $x\neq i$ and $y\neq i$, both of which are negated from the corresponding elements at the analogous indexes of $M'$ such that again this determinant (additive) term is equal to the respective determinant term using $M$.

This shows that term by term within their summed computations, the determinants of minors of $M'$ are everywhere equal to the respective determinants of minors of $M$.  The conclusion is that $M'$ is indeed a $P$-matrix, because $M$ is.
\end{proof}

\noindent For completeness, before continuing we restate \Cref{thm:psdunique} and conclude its proof.

\begin{numberedtheorem}{\ref{thm:psdunique}}[\citealp{GN-65}]
A continuously differentiable function $f : \Omega \to
\reals^n$ with compact and convex product domain $\Omega\subset\reals^n$ is one-to-one if its Jacobian is everywhere a $P$-matrix.
\end{numberedtheorem}
\begin{proof}
The statement now follows as a corollary.  By contradiction, assume there exist distinct $\bomega^1,~\bomega^2\in\Omega$ with $f(\bomega^1)=f(\bomega^2)$.  Let $\mathcal{P}$ be the {\em program} of constraints described in the statement of \Cref{thm:gndiagconc}.  Fix a binary vector $\mathds{B}^{\text{RHS}}\in\left\{\mathds{1},\mathds{-1}\right\}$ to satisfy the right-hand side equations of $\mathcal{P}$ for these $\bomega^1,~\bomega^2$.  Assumptions in the theorem statement here meet the conditions of \Cref{thm:gndiagconc}, therefore $\bomega^1 \neq \bomega^2$ implies that the left-hand side equations of $\mathcal{P}$ can not all be satisfied for $\mathds{B}^{\text{RHS}}$.  In particular it cannot be that $f(\bomega^1) = f(\bomega^2)$ (which would satisfy the left-hand side of $\mathcal{P}$).
\end{proof}

\noindent \Cref{thm:gndiagconc} has implications for our proxy games which become apparent in the proof of \Cref{thm:gameunique}.  The intuition was previously described in the discussion immediately following \Cref{thm:gnpartdiagconc}.

\begin{numberedtheorem}{\ref{thm:gameunique}}
A game with $n$ players and
\begin{itemize}
\item a compact and convex product action space $\Omega_1\times...\times\Omega_n=\Omega\subset\reals^n$; \item a continuous and twice-differentiable utility function $\Utils:\Omega\rightarrow\reals^n$ such that:
\begin{itemize}
\item the pseudogradient $\left[\frac{\partial U_i}{\partial \bomega_i}\right]_i$ of the utility function $\Utils$ is an interior $N$-matrix function;
\item and there exists $\bomega^0\in\Omega$ such that the pseudogradient evaluated at $\bomega^0$ is $\mathbf{0}$ (the 0-vector);
\end{itemize}
\end{itemize}
has a unique Nash equilibrium, which is $\bomega^0$, and this equilibrium is pure.
\end{numberedtheorem}
\begin{proof}
For existence, the action vector $\bomega^0$ is assumed to exist.  It is a pure Nash equilibrium by the assumption that first-order conditions at $\bomega^0$ are all identically 0, and utility functions $\Utili$ are strictly concave with respect to their own unilateral changes (except possibly at the single points of the lower and upper boundaries where it can be weakly concave, but this exception at a single boundary point can not affect the uniqueness of a player's optimal action).  This concavity follows from the pseudogradient as an interior $N$-matrix function, such that at all points (except boundaries), the diagonal elements of the pseudogradient's Jacobian $\left[\frac{\partial^2U_i}{\partial\omega_i^2}\right]_i$ must be strictly negative.

The Nash is unique first because the pseudogradient function is one-to-one by application of \Cref{thm:gnext}, so no other action vector $\bomega'\neq \bomega^0$ can also map to $\mathbf{0}$ (under the pseudogradient function).  Next, the rest of this proof is devoted to showing that a second equilibrium $\bomega'$ can not also exist in the boundaries by having non-zero first-order conditions (i.e., if an agent with action on the boundary has a gradient pointing outside the action space).  An outline is given as follows.
\begin{itemize}
\item First we argue to ignore consideration of any coordinates of $\bomega'$ which are in identified boundaries in their respective dimensions, with respect to the pseudogradient function as the output.  Our goal here is to show that $\bomega'$ and $\bomega^0$ are the same in these coordinates.
\item Second, we consider a dimensionally-reduced function $D$ (\Cref{def:dimreduct}) applied to the pseudogradient function, $\bomega^0,~\bomega'$, and $K = K_{\bomega'}=K_{\bomega^0}$ all the same set of dimensions, where $K_{\bomega'},~K_{\bomega^0}$ are the sets of dimensions in which respectively $\bomega'$ and $\bomega^0$ exist in identified boundaries.  Our arguments under $D$ extend by analogy to our original pseudogradient function if and only if the parameters $\mathbf{c}_K$ used in the definition of $D$ represent the same assignment as the values of the respective coordinates in both $\bomega'$ and $\bomega^0$.
\item Finally, we use $D$ to obtain the contradiction and claim uniqueness of Nash equilibrium.
\end{itemize}

Per the outline, first we show $K_{\bomega'} = K_{\bomega^0}$.  Without loss of generality, we analyze identified lower boundaries, with identified upper boundaries by symmetry.  The simple direction to prove is $K_{\bomega^0} \subseteq K_{\bomega'}$.  By contradiction, assume $\bomegai^0$ is the lower bound of dimension $i$ with its lower boundary identified, but $\omega_i'> \omega_i^0$.  But then $\bomega'$ could not be an equilibrium point, because $\partial U/\partial \omega_i$ outputs 0 everywhere on the lower boundary in dimension $i$ (it is constant on the boundary and we know that it outputs 0 at point $\bomega^0$ by assumption), and $\partial U/\partial \omega_i$ is monotone decreasing in $\omega_i$.

We next show $K_{\bomega'} \subseteq K_{\bomega^0}$, which uses a similar but more technical argument.  Consider $\bomega'$ to be in an identified lower boundary in dimension $i$, with general range $[a_i,b_i]$ for dimension $i$.  The derivative $\partial U_i / \partial \omega_i$ is the element of the (output) psuedogradient function with index $i$.  By definition of an identified boundary, the output of $\partial U_i / \partial \omega_i$ is constant for inputs $(a_i,\bomegasmi)$ for all $\bomegasmi$.  At the lower boundary, it can not be that $\partial U_i / \partial \omega_i(a_i,\bomegasmi') > 0$ without contradicting $\bomega'$ as an equilibrium, so it must be that $\partial U_i / \partial \omega_i(a_i,\bomegasmi') \leq 0$.

However because $\partial U_i / \partial \omega_i(a_i,\cdot)$ is constant, then it must also be that $\partial U_i / \partial \omega_i(a_i,\bomegasmi^0) \leq 0$, which implies that $\bomega^0$ must also have $\omega^0_i = a_i$ (and in fact by assumption $\partial U_i / \partial \omega_i(a_i=\omega^0_i,\bomegasmi^0) = 0$).  This follows because any other (larger) value of $\omega^0_i$ would contradict $\bomega^0$ as an equilibrium, from the non-positive derivative at $a_i$ and the strict concavity at all interior points from the pseudogradient being an interior $N$-matrix function.

So we have $\omega^0_i = \omega'_i = a_i$ and $\partial U_i / \partial \omega_i(a_i,\bomegasmi^0)=\partial U_i / \partial \omega_i(a_i,\bomegasmi') = 0$.  The intermediate conclusion here is that dimension $i$ 
can not be used to maintain that $\bomega'$ is distinct from $\bomega^0$.  
Further, the analysis so far has applied for general $i$.  Therefore, it must be that for every dimension $i$ where $\bomega'$ is in an identified boundary in dimension $i$, $\bomega^0$ must be in each of the same identified boundaries; i.e., it must be that $K_{\bomega'} \subseteq K_{\bomega^0}$.  

We continue to the second point of the outline.  From this point on, we consider the dimensionally-reduced function $D$ applied to the pseudogradient function, $\bomega^0, \bomega'$ and $K$ the (common) set of dimensions for which $\bomega'$ and $\bomega^0$ each exist in 
identified boundaries.  The reduction to $D$ in space $\Omega_{-K}$ is faithful for the following analysis because the coordinates of the pseudogradient fixed by $\mathbf{c}_K$ reflect both $\bomega'$ and $\bomega^0$.  Putting together the definitions of a dimensionally-reduced function (applied to $D$) and interior $N$-matrix function (applied to the pseudogradient), we have that $D$ is a strict $N$-matrix function, i.e., its Jacobian is an $N$-matrix everywhere on its (reduced) domain.

We now prove a contradiction.  By the contrapositive of (the $N$-matrix version of) \Cref{thm:gndiagconc}, for the pseudogradient function and our two distinct inputs, there must exist at least one dimension $i$ such that
\begin{align*}
    (\omega'_i - \omega^0_i)\cdot\left(\frac{\partial U_i}{\partial \omega_i}(\bomega')-\frac{\partial U_i}{\partial \omega_i}(\bomega^0)\right) &< 0 \\
    \Leftrightarrow \quad (\omega'_i - \omega^0_i)\cdot \frac{\partial U_i}{\partial \omega_i}(\bomega') &< 0
\end{align*}
where the second line drops the derivative at $\bomega^0$ because it is known to be 0.

If $\omega'_i<\omega^0_i$, it must be that the pseudogradient at $\bomega'$ in dimension $i$ is greater than 0; alternatively if $\omega'_i>\omega^0_i$, this pseudogradient element is less than 0.  But both cases contradict $\bomega'$ as a Nash point because in either case, the gradient points back in the direction of $\bomega^0$, and the action space is convex which therefore guarantees that player $i$ has a better response than $\omega'_i$ when others play $\bomegasmi'$.
\end{proof}


\begin{numberedcorollary}{\ref{thm:nmunique}}
A game with a compact and convex product action space and pseudogradient that is an $N$-matrix function has a unique Nash equilibrium, which is pure.
\end{numberedcorollary}
\begin{proof}
The description of the game here is sufficient to meet the conditions of \Cref{thm:rosenupcontour} (below) from \citep{ros-65}, with existence of pure Nash gauranteed as a result.  Intuitively, existence of pure Nash follows from the combination of continuity of the utility functions and resulting continuity of upper-countour sets, and fixed point theorems on compact and convex spaces.

The intuition for uniqueness is that it follows from \Cref{thm:gndiagconc}, with structure and explanation mostly analogous to the proof of \Cref{thm:gameunique}.  In contrast to \Cref{thm:gameunique} however, because we have a strict $N$-matrix function as the pseudogradient, we do not need to make special arguments regarding identified boundaries.

Formally we argue uniqueness by contradiction.  Assume there exist two distinct pure Nash equilibrium points.  \Cref{thm:psdunique} says there exists a bijection between action space and the image of the pseudogradient function on utility (with the action space as domain).  But \Cref{thm:gndiagconc} requires that there must exist a dimension in which one of the two supposed-distinct equilibrium points has a gradient pointing strictly in the direction of the other, a contradiction because the action space is convex so a preferred deviation much exist.
\end{proof}

\noindent For completeness we give the theorem by \citeauthor{ros-65} referenced in \Cref{thm:nmunique}.

\begin{theorem}[\citealp{ros-65}]
\label{thm:rosenupcontour}
Consider a game with $n$ players and a compact and convex product action space $\Omega$.  Assume the utility function $\Utils$ is continuous and for every player $i$ and vector of others actions' $\bomegasmi$, the function $\Utili(\omega_i,\bomegasmi)$ is concave in $\omega_i$.  There exists a pure Nash equilibrium.
\end{theorem}





\subsection{Derivative Calculations; Proofs of \Cref{lem:pderi}, \Cref{lem:concutil}, \Cref{lem:priceid}}
\label{a:pderi}
\ifsodacut
\else
\ifsoda
\begin{align*}
\intertext{Allocation rule sub-calculations:}
\alloci(\wals) &= \frac{\wali}{\sum_k \wali[k]}\\
\frac{\partial\alloci}{\partial \wali}(\wals) &= \frac{\left(\sum_k \wali[k]\right)-\wali}{\left(\sum_k \wali[k]\right)^2}\\
\frac{\partial\alloci}{\partial \wali[j]}(\wals) &= \frac{-\wali}{\left(\sum_k \wali[k]\right)^2} = \frac{-\alloci(\wals)}{\sum_k \wali[k]}\\
\intertext{Re-stating the bid function:}
\payi(\wals)&=v_i(\wali)-\frac{\int^{\wali}_{\wali(0)}\alloci(z,\walsmi)\vali'(z)dz}{
    \alloci(\wals)}
\end{align*}
\begin{align*}
\intertext{Self-partial:}
&\frac{\partial \payi}{\partial \wali}(\wals)=\vali'(\wali)-\frac{\alloci(\wals)\vali'(\wali)}{\alloci(\wals)}\\
&\quad+\frac{\int^{\wali}_{\wali(0)}\alloci(z,\walsmi)\vali'(z)dz\cdot \frac{\partial \alloci}{\partial \wali}(\wals)}{\alloci^2(\wals)}\\
&=\frac{\int^{\wali}_{\wali(0)}\alloci(z,\walsmi)\vali'(z)dz\cdot \frac{\partial \alloci}{\partial \wali}(\wals)}{\alloci^2(\wals)}\\
&=\frac{\int^{\wali}_{\wali(0)}\alloci(z,\walsmi)\vali'(z)dz\cdot \frac{\left(\sum_k \wali[k]\right)-\wali}{\left(\sum_k \wali[k]\right)^2}}{\left(\frac{\wali}{\sum_k \wali[k]}\right)^2}\\
&= \int^{\wali}_{\wali(0)}\alloci(z,\walsmi)\vali'(z)\left[ \frac{\left(\sum_k \wali[k]\right)-\wali}{\wali^2} \right]dz\\
&= \int^{\wali}_{\wali(0)}\vali'(z)\frac{1}{\wali}\cdot \frac{z}{\left(\sum_k \wali[k]\right)-\wali+z}\left[\frac{\sum_k \wali[k]}{\wali}-1\right]dz
\end{align*}
\footnotesize
\begin{align*}
\intertext{Cross-partials:}
&\frac{\partial \payi}{\partial \wali[j]}(\wals)=
-\frac{\int^{\wali}_{\wali(0)}\frac{\partial \alloci}{\partial \wali[j]}(z,\walsmi)\vali'(z)dz}{\alloci(\wals)}\\
&\quad+ \frac{\frac{\partial \alloci}{\partial \wali[j]}(\wals)\int^{\wali}_{\wali(0)}\alloci(z,\walsmi)\vali'(z)dz}{\alloci^2(\wals)}\\
&= \frac{\int^{\wali}_{\wali(0)}\vali'(z)\left[\alloci(z,\walsmi)\frac{\partial \alloci}{\partial \wali[j]}(\wals)
-\frac{\partial \alloci}{\partial \wali[j]}(z,\walsmi)\alloci(\wals)\right]dz}{\alloci^2(\wals)}\\
&= \frac{\int^{\wali}_{\wali(0)}\vali'(z)\left[\alloci(z,\walsmi)\frac{-\alloci(\wals)}{\sum_k \wali[k]}
-\frac{-\alloci(z,\walsmi)}{\left(\sum_k \wali[k]\right)-\wali+z}\alloci(\wals)\right]dz}{\alloci^2(\wals)}\\
&= \frac{\int^{\wali}_{\wali(0)}\vali'(z)\left[\alloci(z,\walsmi)\frac{-1}{\sum_k \wali[k]}
+\frac{\alloci(z,\walsmi)}{\left(\sum_k \wali[k]\right)-\wali+z}\right]dz}{\frac{\wali}{\sum_k \wali[k]}}\\
&= \int^{\wali}_{\wali(0)}\frac{\vali'(z)}{\wali}\cdot \frac{z}{\left(\sum_k \wali[k]\right)-\wali+z}\left[\frac{\sum_k \wali[k]}{\left(\sum_k \wali[k]\right)-\wali+z}-1\right]dz
\end{align*}
\normalsize
\else
\begin{align*}
\intertext{Allocation rule sub-calculations:}
\alloci(\wals) &= \frac{\wali}{\sum_k \wali[k]}\\
\frac{\partial\alloci}{\partial \wali}(\wals) &= \frac{\left(\sum_k \wali[k]\right)-\wali}{\left(\sum_k \wali[k]\right)^2}\\
\frac{\partial\alloci}{\partial \wali[j]}(\wals) &= \frac{-\wali}{\left(\sum_k \wali[k]\right)^2} = \frac{-\alloci(\wals)}{\sum_k \wali[k]}\\
\intertext{Re-stating the bid function:}
\payi(\wals)&=v_i(\wali)-\frac{\int^{\wali}_{\wali(0)}\alloci(z,\walsmi)\vali'(z)dz}{
    \alloci(\wals)}\\
\intertext{Self-partial:}
\frac{\partial \payi}{\partial \wali}(\wals)&=\vali'(\wali)-\frac{\alloci(\wals)\vali'(\wali)}{\alloci(\wals)}+\frac{\int^{\wali}_{\wali(0)}\alloci(z,\walsmi)\vali'(z)dz\cdot \frac{\partial \alloci}{\partial \wali}(\wals)}{\alloci^2(\wals)}\\
&=\frac{\int^{\wali}_{\wali(0)}\alloci(z,\walsmi)\vali'(z)dz\cdot \frac{\partial \alloci}{\partial \wali}(\wals)}{\alloci^2(\wals)}\\
&=\frac{\int^{\wali}_{\wali(0)}\alloci(z,\walsmi)\vali'(z)dz\cdot \frac{\left(\sum_k \wali[k]\right)-\wali}{\left(\sum_k \wali[k]\right)^2}}{\left(\frac{\wali}{\sum_k \wali[k]}\right)^2}\\
&= \int^{\wali}_{\wali(0)}\alloci(z,\walsmi)\vali'(z)\left[ \frac{\left(\sum_k \wali[k]\right)-\wali}{\wali^2} \right]dz\\
&= \int^{\wali}_{\wali(0)}\vali'(z)\frac{1}{\wali}\cdot \frac{z}{\left(\sum_k \wali[k]\right)-\wali+z}\cdot\left[\frac{\sum_k \wali[k]}{\wali}-1\right]dz\\
\intertext{Cross-partials:}
\frac{\partial \payi}{\partial \wali[j]}(\wals)&=
-\frac{\int^{\wali}_{\wali(0)}\frac{\partial \alloci}{\partial \wali[j]}(z,\walsmi)\vali'(z)dz}{\alloci(\wals)}
+ \frac{\frac{\partial \alloci}{\partial \wali[j]}(\wals)\int^{\wali}_{\wali(0)}\alloci(z,\walsmi)\vali'(z)dz}{\alloci^2(\wals)}\\
&= \frac{\int^{\wali}_{\wali(0)}\vali'(z)\cdot\left[\alloci(z,\walsmi)\frac{\partial \alloci}{\partial \wali[j]}(\wals)
-\frac{\partial \alloci}{\partial \wali[j]}(z,\walsmi)\alloci(\wals)\right]dz}{\alloci^2(\wals)}\\
&= \frac{\int^{\wali}_{\wali(0)}\vali'(z)\cdot\left[\alloci(z,\walsmi)\frac{-\alloci(\wals)}{\sum_k \wali[k]}
-\frac{-\alloci(z,\walsmi)}{\left(\sum_k \wali[k]\right)-\wali+z}\alloci(\wals)\right]dz}{\alloci^2(\wals)}\\
&= \frac{\int^{\wali}_{\wali(0)}\vali'(z)\cdot\left[\alloci(z,\walsmi)\frac{-1}{\sum_k \wali[k]}
+\frac{\alloci(z,\walsmi)}{\left(\sum_k \wali[k]\right)-\wali+z}\right]dz}{\frac{\wali}{\sum_k \wali[k]}}\\
&= \int^{\wali}_{\wali(0)}\vali'(z)\frac{1}{\wali}\cdot\left[\frac{-z}{\left(\sum_k \wali[k]\right)-\wali+z}+\frac{z}{(\left(\sum_k \wali[k]\right)-\wali+z)^2}\cdot\sum_k \wali[k]\right]dz\\
&= \int^{\wali}_{\wali(0)}\vali'(z)\frac{1}{\wali}\cdot \frac{z}{\left(\sum_k \wali[k]\right)-\wali+z}\cdot\left[\frac{\sum_k \wali[k]}{\left(\sum_k \wali[k]\right)-\wali+z}-1\right]dz
\end{align*}
\fi
\fi

\begin{numberedlemma}{\ref{lem:pderi}}
Given the price function $\pays$ for proportional weights, for $j,k\neq i$, the cross derivatives are the same: $\frac{\partial \payi}{\partial w_j} =\frac{\partial \payi}{\partial w_k}$.  Evaluating the Jacobian at $\wals$, further, all elements of the Jacobian matrix $J_{\pays}$ are positive, i.e., $\frac{\partial \payi}{\partial \wali}>0,~\frac{\partial \payi}{\partial w_j}>0$, except at the $\wali(0)$ lower boundary in dimension $i$ where the elements of row $i$ are $\frac{\partial \payi}{\partial \wali}= \frac{\partial \payi}{\partial w_j}=0$.  The lower boundaries are identified.
\end{numberedlemma}
\ifsodacut
\else
\begin{proof}
All cross-derivatives $\frac{\partial \payi}{\partial \wali[j]}$ for fixed $i$ and $j\neq i$ are equal because a $d\wali[j]$ increase in the weight of any other agent $j$ ``looks the same" mathematically to the proportional weights allocation rule of agent $i$, which is $\alloci(\wali) = \frac{\wali}{\wali+\sum_{j\neq i} \wali[j]}$.

We continue by recalling our assumption that weights are strictly positive and strictly increasing in value.  Then all terms in the derivative equations~\eqref{eqn:partialsame} and~\eqref{eqn:partialcross} within the integrals are non-negative everywhere by inspection.  All denominator terms are strictly positive everywhere.

For any dimension $i$, consider $\wali>\wali(0)$.  
For integrand $z$ strictly interior to the endpoints in $\left(\wali(0),\wali\right)$, all terms in the derivative equations are strictly positive everywhere.  With non-negativity everywhere and positivity somewhere, all derivatives evaluate to be strictly positive.
\end{proof}
\fi

\begin{lemma}
\label{lem:concutil}
Each function $\cumimbali$ is strictly concave taking derivatives with respect to $i$, except at the lower end point of its domain where it is weakly concave.
\end{lemma}
\ifsodacut
\else
\begin{proof}
We have $\cumimbali(\ewali,\ewalsmi) = \int_{\wali(0)}^{\ewali} \imbali(z,\ewalsmi)\,\dd z = \int_{\wali(0)}^{\ewali}\rezi - \payi(\ewali,\ewalsmi)\,\dd z $, i.e., the function $\cumimbali$ is defined as the integral over the quantity which subtracts the price function $\payi$ from a constant price term $\rezi$.  In \Cref{lem:pderi} (appearing immediately above), function $\payi$ is shown to be monotone strictly increasing on its domain except at the lower bound where its derivative is 0.  Such an integral is concave on its domain as stated.
\end{proof}
\fi


\begin{lemma}
\label{lem:priceid}
Given agents with (unknown) values $\vals\in\left[0,\maxval\right]^n$.  Consider the price function $\pays$ resulting from a strictly increasing, continuous and differentiable proportional weights social choice function $\allocs$, and dominant-strategy incentive-compatible mechanism implementing $\allocs$.  The lower boundaries of weights space are identified boundaries (\Cref{def:identbound}).
\end{lemma}
\ifsodacut
\else
\begin{proof}
By our assumptions in \Cref{s:prelim} for a proportional weights social choice function $\allocs$, its (parameter) weight functions are strictly positive, even for an agent with value 0.  Self-partials in equation~\eqref{eqn:partialsame} and cross-partials in equation~\eqref{eqn:partialcross} are well-defined.  By \Cref{lem:pderi}, for each $i$ the cross-partials at the lower bound of weight space $\wali(0)$ are everywhere identically 0, for all $i$, regardless of $\walsmi$, meeting the first requirement in the definition of an identified boundary.  Again by \Cref{lem:pderi}, self-partials $\partial f_i/\partial \omega_i$ are strictly positive everywhere above the lower boundary ($d_i > a_i$, c.f. proof of \Cref{lem:identify}): these are the diagonal element of the Jacobian at index $(i,i)$.  Therefore $f_i(d_i,\bomega_{-i})=f_i(a_i,\bomega_{-i})+\int_{a_i}^{d_i}\left[\partial f_i /\partial \omega_i(z,~\bomegasmi) \right]dz>f_i(a_i,\bomega_{-i})$.
\end{proof}
\fi

\subsection{Proof of \Cref{thm:mpd} in Section~\ref{s:oneweightspsd}}
\label{as:hmatrixpsd}

\begin{numberedtheorem}{\ref{thm:mpd}}
Consider a $K\times K$ matrix $\scrbh$ with diagonal $\scrlh_1, \scrlh_2, \ldots,
\scrlh_K$ and all other entries equal to 1 (and without loss of
generality $\scrlh_1 \leq \scrlh_2 \leq \ldots \leq \scrlh_K$).  
The following is a complete characterization describing when $\scrbh$ is
positive definite.
\begin{enumerate}
\item if $\scrlh_1 \leq 0$, then the matrix $\scrbh$ is not positive definite;
\item if $\scrlh_1 \geq 1$ and $\scrlh_2 > 1$, then $\scrbh$ is positive definite;
\item if $0<\scrlh_1,\scrlh_2 \leq 1$, then $\scrbh$ is not positive definite;
\item if $0<\scrlh_1< 1$ and $\scrlh_2>1$, then $\scrbh$ is positive definite if and only if $\sum_k \frac{1}{1-\scrlh_k} > 1$.
\end{enumerate}
\end{numberedtheorem}

\ifsodacut
\else
\begin{proof}
To prove positive definiteness in cases (2) and (4), we will show that for any non-zero vector $\vecbz$, it must be true that $\vecbz^\top \scrbh\,\vecbz>0$.  For cases (1) and (3) we give counterexamples of $\vecbz$ for which $\vecbz^\top \scrbh\,\vecbz\leq 0$.  Given the structure of $\scrbh$ (as all ones except the diagonal), we have
\begin{equation}
\label{eqn:zmz}
\vecbz^\top \scrbh\,\vecbz=\left(\sum\nolimits_i z_i\right)^2+\sum\nolimits_i(\scrlh_i-1)z_i^2.
\end{equation}

We recall for use throughout this proof the assumption that, without
loss of generality, the diagonal elements are such that $\scrlh_1 \leq
\scrlh_2 \leq \ldots \leq \scrlh_K$.  We prove each case of the
characterization in turn.

Case (1) is correct by counter-example, setting $\vecbz =
(-1,0,\ldots,0)$.\footnote{Of course, it is a well-known property of
  positive definite matrices $\scrbh$ that all diagonal elements must
  be strictly positive, otherwise they have $\vecbz^\top
  \scrbh\,\vecbz \leq 0$ with a simple counter-example $\vecbz$
  described by all zeroes except $-1$ in the index of the matrix's
  non-positive diagonal element.}

Case (2) is correct by inspection of equation~\eqref{eqn:zmz} in which all terms are non-negative.  The vector $\vecbz$ is non-zero, so either a $(\scrlh_j-1)z_j^2$ term for $j\neq1$ in the second sum is strictly larger than 0,  or all such $z_j$ are 0 but then $z_1\neq 0$ and the first sum-squared is strictly larger than 0.

Case (3) is correct by counter-example, setting $\vecbz = (1,-1,0,\ldots,0)$.

For case (4), we need to prove that when $0<\scrlh_1< 1$ and $\scrlh_2>1$, then the matrix $\scrbh$ is positive definite if and only if $\sum_k \frac{1}{1-\scrlh_k} > 1$.

For this last case, given the assumptions on the $\scrlh_i$ elements, only the $(\scrlh_1-1)z_1^2$ term from equation~\eqref{eqn:zmz} is negative, all other terms are non-negative.  Therefore, from this point on, we can ignore any sub-case where $z_1 = 0$, as some $(\scrlh_j-1)z_j^2$ term for $j\neq 1$ must be strictly positive.

Now consider fixing the value $z_1$ to any real number $\zbar\neq 0$.
We will show that equation~\eqref{eqn:zmz} is strictly positive for any
$\zees \in\reals^{n-1}$.  Specifically, for any fixed $\zbar\neq 0$,
equation~\eqref{eqn:zmz} has a global minimum in variables $\zees$ that is
strictly positive.  This global minimum $\optzees$ satisfies
\ifsoda
\begin{align}
&\optzees =\argmin_{\zees}(\bar{z}_1,\vecbz_{-1})^{\top}\cdot \scrbh\cdot(\bar{z}_1,\vecbz_{-1})\\
\label{eqn:hmatrixtestobjective}
&= \argmin_{\zees}
\left(\bar{z}_1+\sum\nolimits_{j\geq2} z_j\right)^2+\sum\nolimits_{j\geq2}(\scrlh_j-1)z_j^2
\end{align}
\else
\begin{align}
\optzees &=\argmin_{\zees}(\bar{z}_1,\vecbz_{-1})^{\top}\cdot \scrbh\cdot(\bar{z}_1,\vecbz_{-1})\\
\label{eqn:hmatrixtestobjective}
&= \argmin_{\zees}
\left(\bar{z}_1+\sum\nolimits_{j\geq2} z_j\right)^2+\sum\nolimits_{j\geq2}(\scrlh_j-1)z_j^2
\end{align}
\fi
where the second line substitutes equation~\eqref{eqn:zmz} and drops
the constant $\zbar$ term from the right hand sum.  It will be
convenient to denote the sum of the variables as $\zum = \zbar +
\sum_{j\geq 2} \optzeei[j]$.  After the brief argument that the minimizer
$\optzees$ exists and is characterized by its first-order conditions,
we will use first-order conditions on $\optzees$ to write all
variables in terms of $\zum$ which we substitute into \eqref{eqn:zmz}
to analyze.

To show that $\optzees$ exists and is characterized by its first-order
conditions, observe that the polynomial $(\bar{z}_1,\vecbz_{-1})^{\top}\scrbh\,(\bar{z}_1,\vecbz_{-1})$ is a quadratic form with Hessian 
$2\cdot\scrbh_{\left[2:K,2:K\right]}$, i.e., twice the
matrix $\scrbh$ without the first row and column:\ifsoda
~$\text{Hessian}((\bar{z}_1,\vecbz_{-1})^{\top}\scrbh\,(\bar{z}_1,\vecbz_{-1}) )=$
\[
\scrbh_{\left[2:K,2:K\right]}= 
	\begin{bmatrix}
  \scrlh_2  &     1 & \dots  &    1 \\
   1  &  \scrlh_3 & \dots  &    1 \\
 \vdots & \vdots & \ddots & \vdots \\
    1  &      1  & \dots  &      \scrlh_K
\end{bmatrix}.
\]
\else
\[
	\text{Hessian}((\bar{z}_1,\vecbz_{-1})^{\top}\scrbh\,(\bar{z}_1,\vecbz_{-1}) )=\scrbh_{\left[2:K,2:K\right]}=\begin{bmatrix}
  \scrlh_2  &     1 & \dots  &    1 \\
   1  &  \scrlh_3 & \dots  &    1 \\
 \vdots & \vdots & \ddots & \vdots \\
    1  &      1  & \dots  &      \scrlh_K
\end{bmatrix}.
\]
\fi
Matrix $\scrbh_{\left[2:K,2:K\right]}$ is ones except by assumption we
have $\scrlh_j>1$ for $j\geq 2$ in the diagonal; thus, by case (2) of
the theorem, it is positive definite.  A quadratic
form with strictly positive definite Hessian has a unique local minimum which is characterized by its
first-order conditions.

We now use the first-order conditions to write optimizer $\optzees$
of equation~\eqref{eqn:hmatrixtestobjective} in terms of $\zum$.\footnote{Note that line~\eqref{eqn:uniquefocs} is not a definition for $z_j^*$, which appears on both sides of the equation.  The goal is substitution of $z_j^*$ from necessary first-order conditions, not to define it.}
\ifsoda
\begin{align}
\intertext{The following two lines each hold for each $j\geq 2$:}
\label{eqn:firstfocsdef}
0 &= 2\left(\bar{z}_1 + \left(\sum\nolimits_{k \geq 2, k\neq j} z_k^*\right)+(\scrlh_j-1)\cdot z_j^*\right) \\
\intertext{and re-arranging:}
\label{eqn:uniquefocs}
z_j^*&=\frac{1}{1-\scrlh_j}\,\zum 
\end{align}
\else
\begin{align}
\label{eqn:firstfocsdef}
0 &= 2\left(\bar{z}_1 + \left(\sum\nolimits_{k \geq 2, k\neq j} z_k^*\right)+(\scrlh_j-1)\cdot z_j^*\right) &\text{for each $j\geq 2$}\\
\intertext{and re-arranging:}
\label{eqn:uniquefocs}
z_j^*&=\frac{1}{1-\scrlh_j}\,\zum &\text{for each $j\geq 2$}
\end{align}
\fi

We now similarly identify a substitution of $\zbar$ in terms of
$\zum$.  Starting from equation~\eqref{eqn:uniquefocs}, sum the
$z_j^*$ first-order condition equalities over all $j\geq 2$:
\begin{align}
\sum_{j\geq2} z_j^*&=\sum_{j\geq2}\left(\frac{1}{1-\scrlh_j} \zum\right)\\
\intertext{Add $\zbar$ to both sides of the equation:}
1\cdot\left(\bar{z}_1+\sum\nolimits_{j\geq2} z_j^*\right)&=\bar{z}_1+\left(\sum\nolimits_{j\geq2}\frac{1}{1-\scrlh_j}\right)\zum\\
\intertext{Substitute $\zum$ on the left and solve for the right-hand side $\zbar$ term:}
\label{eqn:zone}
\bar{z}_1&=\left(1-\sum\nolimits_{j\geq 2} \frac{1}{1-\scrlh_j}\right)\cdot\zum.
\end{align}
Notice that equation~\eqref{eqn:zone} and the definition of $\zbar \neq
0$ excludes the possibility that $\zum=0$.

In the analysis below, the first line re-writes the objective function
in~\eqref{eqn:zmz}.
The second line substitutes equations~\eqref{eqn:uniquefocs}
and~\eqref{eqn:zone}.  Subsequent lines are elementary manipulations.\ifsoda  For space considerations, we divide up front by $\zum^2$ and re-summarize at the end.

\begin{align*}
&\frac{1}{\zum^2}(\zbar,\optzees)^{\top}\cdot \scrbh\cdot(\zbar,\optzees)\\ 
&= \frac{1}{\zum^2}\left[(\scrlh_1-1)\,\zbar + \zum^2 + \sum\nolimits_{j\geq 2} \left(\scrlh_j - 1\right)\,\optzeei[j] \right]\\
&=\left(\scrlh_1-1\right)\left(1-\sum_{j\geq2}\frac{1}{1-\scrlh_j}\right)^2+
\left(1-\sum_{j\geq2}\frac{1}{1-\scrlh_j}\right)\\
&=\left[\left(\scrlh_1-1\right)\left(1-\sum_{j\geq2}\frac{1}{1-\scrlh_j}\right)^2+1-\sum_{j\geq2}\frac{1}{1-\scrlh_j}\right]\\
&=\left(1-\sum\nolimits_{j\geq2}\frac{1}{1-\scrlh_j}\right)\left(1-\scrlh_1\right)\left[\sum\nolimits_k\frac{1}{1-\scrlh_k}-1\right]
\intertext{such that}
&(\zbar,\optzees)^{\top}\cdot \scrbh\cdot(\zbar,\optzees)=\\
&\zum^2\left(1-\sum\nolimits_{j\geq2}\frac{1}{1-\scrlh_j}\right)\left(1-\scrlh_1\right)\left[\sum\nolimits_k\frac{1}{1-\scrlh_k}-1\right]
\end{align*}
\else
\begin{align*}
\MoveEqLeft
(\zbar,\optzees)^{\top}\cdot \scrbh\cdot(\zbar,\optzees)\\ 
&= (\scrlh_1-1)\,\zbar + \zum^2 + \sum\nolimits_{j\geq 2} \left(\scrlh_j - 1\right)\,\optzeei[j] \\
&=\left(\scrlh_1-1\right)\left(1-\sum_{j\geq2}\frac{1}{1-\scrlh_j}\right)^2\zum^2+
\left(1-\sum_{j\geq2}\frac{1}{1-\scrlh_j}\right)\zum^2\\
&=\zum^2\,\left[\left(\scrlh_1-1\right)\left(1-\sum_{j\geq2}\frac{1}{1-\scrlh_j}\right)^2+1-\sum_{j\geq2}\frac{1}{1-\scrlh_j}\right]\\
&=\zum^2\,\left(1-\sum_{j\geq2}\frac{1}{1-\scrlh_j}\right)\left[\left(1-\scrlh_1\right)\left(\sum_{j\geq2}\frac{1}{1-\scrlh_j}-1\right)+\frac{1-\scrlh_1}{1-\scrlh_1}\right]\\
&=\zum^2\,\left(1-\sum\nolimits_{j\geq2}\frac{1}{1-\scrlh_j}\right)\left(1-\scrlh_1\right)\left[\sum\nolimits_k\frac{1}{1-\scrlh_k}-1\right].
\end{align*}
\fi
Given the assumptions on the $\scrlh_i$ for current case (4), the
first three terms of this product are strictly positive (recalling $0<\scrlh_1<1$ and $\scrlh_j>1$ for $j>1$, and $\zbar \neq 0$ and $\zum\neq 0$, so $(\zum)^2>0$).  

To finish, we
observe that the exact dependence of positive definiteness of the matrix
$\scrbh$ is on the bracketed fourth term (where the first term $k=1$
of the sum is positive and all of the other terms are negative):
\ifsoda
\begin{align*}
\label{eqn:lastline}
&\text{for}~0<\scrlh_1<1 \textrm{ and } \scrlh_j > 1~\forall j\geq 2,~\scrbh \text{ is positive definite }\\
&\textit{iff } \left[\sum\nolimits_k\tfrac{1}{1-\scrlh_k}-1\right] > 0. \qedend
\end{align*}
\else
\begin{equation*}
\label{eqn:lastline}
\textrm{ For } 0<\scrlh_1<1 \textrm{ and } \scrlh_j > 1~\forall j\geq 2,~\scrbh \text{ is positive definite } \textit{iff } \left[\sum\nolimits_k\tfrac{1}{1-\scrlh_k}-1\right] > 0. \qedend
\end{equation*}
\fi
\end{proof}
\fi

\subsection{Lemmas Supporting Theorem~\ref{thm:mpd2} in Section~\ref{s:oneweightspsd}}
\label{as:thmmpd}

\begin{numberedlemma}{\ref{lem:halfw}}
If $h_i\leq 1$, then $w_i> 0.5\sum_k w_k$, and all other weights must have $w_j < 0.5\sum_k w_k$, and all other $h_j>1$.
\end{numberedlemma}
\ifsodacut
\else
\begin{proof}
Writing out $h_i$ from its definition as the ratio of partial derivatives,\ifsoda~$h_i=$
\begin{equation*}
\frac{ \int^{\wali}_{\wali(0)}\vali'(z)\frac{1}{\wali}\cdot \frac{z}{\sum_k \wali[k]-\wali+z}\cdot\left[\frac{\sum_k \wali[k]}{\wali}-1\right]dz}{ \int^{\wali}_{\wali(0)}\vali'(z)\frac{1}{\wali}\cdot \frac{z}{\sum_k \wali[k]-\wali+z}\cdot\left[\frac{\sum_k \wali[k]}{\sum_k \wali[k]-\wali+z}-1\right]dz}
\end{equation*}
\else
\begin{eqnarray*}
&&h_i=\frac{ \int^{\wali}_{\wali(0)}\vali'(z)\frac{1}{\wali}\cdot \frac{z}{\sum_k \wali[k]-\wali+z}\cdot\left[\frac{\sum_k \wali[k]}{\wali}-1\right]dz}{ \int^{\wali}_{\wali(0)}\vali'(z)\frac{1}{\wali}\cdot \frac{z}{\sum_k \wali[k]-\wali+z}\cdot\left[\frac{\sum_k \wali[k]}{\sum_k \wali[k]-\wali+z}-1\right]dz}
\end{eqnarray*}
\fi
If $h_i\leq 1$, by implication it is well-defined so the denominator can not disappear and $\wali>\wali(0)$.  There must exist $ z\in\left(0,\wali\right]$, such that

\begin{eqnarray}
\label{eqn:smallhi}
\frac{\sum_k w_k}{\sum_k w_k-w_i+z}\geq\frac{\sum_k w_k}{w_i}
\end{eqnarray}

\noindent which implies $w_i> 0.5\sum_k w_k$ by noting equal numerators and comparison of denominators.  
The rest of the claim follows as $\wali$ is obviously the only weight more than half the total, and claiming $h_j>1$ for other $j$ is simply an explicit statement of the contrapositive.
\end{proof}
\fi

\Cref{lem:fin} proves the necessary and sufficient lower bound to show that $\pays$ meets the conditions of Theorem~\ref{thm:mpd2} Case (4).  Technical \Cref{lem:perterm} below it is used by \Cref{lem:fin}.

\begin{numberedlemma}{\ref{lem:fin}}
When $h_1<1$ and $h_j>1~\forall j\neq1$, we have $\sum_k \frac{1}{1-h_k}>1$.
\end{numberedlemma}
\ifsodacut
\else
\begin{proof}
\label{proof:fin}
With $h_1<1$ by assumption, then $\wali[1]>0.5\sum_k \wali[k]$ by Lemma \ref{lem:halfw}, and $\alloci[1]>0.5$.
Thus $\alloci[j]<0.5$ for $j\neq 1$ and we can apply Lemma \ref{lem:perterm} (below), to get the first inequality in the following analysis:
\begin{eqnarray*}
\sum_k \frac{1}{1-h_k}&>&\frac{\alloci[1]^2}{2\alloci[1]-1}+\sum_{k>1}\frac{\alloci[k]^2}{2\alloci[k]-1}\\
&\geq&\frac{\alloci[1]^2}{2\alloci[1]-1}+\frac{(1-\alloci[1])^2}{2(1-\alloci[1])-1}\\
&=&1
\end{eqnarray*}

\noindent and with the second step following because $\frac{\alloci[k]^2}{2\alloci[k]-1}\Big|_0=0$ and is a concave function when $0<\alloci[k]<0.5$ and $\sum_{k>1} \alloci[k] = (1-\alloci[1])$  
(its second derivative is $\frac{2}{(2\alloci[k]-1)^3}$ and it acts submodular).
\end{proof}
\fi

\begin{lemma}
When $h_1<1$ and $h_j>1~\forall j\neq1$, then $\forall~i\in\left\{1,\ldots,n\right\}$, we have 
\label{lem:perterm}
$\frac{1}{1-h_i}>\frac{\alloci^2}{2\alloci-1}$.
\end{lemma}
\ifsodacut
\else
\begin{proof}
By subtracting 1 from both sides, it is equivalent to prove the inequality on the right:
\ifsoda
$$\frac{1}{1-h_i}>\frac{\alloci^2}{2\alloci-1}
 \Leftrightarrow \frac{h_i}{1-h_i}>\frac{\alloci^2-2\alloci+1}{2\alloci-1}=\frac{(1-\alloci)^2}{2\alloci-1}$$
\else
$$\frac{1}{1-h_i}>\frac{\alloci^2}{2\alloci-1}
\qquad \Longleftrightarrow\qquad \frac{h_i}{1-h_i}>\frac{\alloci^2-2\alloci+1}{2\alloci-1}=\frac{(1-\alloci)^2}{2\alloci-1}$$
\fi
Working from the definition of $h_i$:
\ifsoda
\begin{align*}
&\frac{h_i}{1-h_i}=\\
&\frac{ \int^{\wali}_{\wali(0)}\vali'(z)\frac{1}{\wali}\cdot \frac{z}{\sum_k \wali[k]-\wali+z}\cdot\left[\frac{\sum_k \wali[k]}{\wali}-1\right]dz}
{ \int^{\wali}_{\wali(0)}\vali'(z)\frac{1}{\wali}\cdot \frac{z}{\sum_k \wali[k]-\wali+z}\cdot\left[\frac{\sum_k \wali[k]}{\sum_k \wali[k]-\wali+z}-\frac{\sum_k \wali[k]}{\wali}\right]dz}
\end{align*}
\else
\begin{eqnarray*}
\frac{h_i}{1-h_i}&=&
\frac{ \int^{\wali}_{\wali(0)}\vali'(z)\frac{1}{\wali}\cdot \frac{z}{\sum_k \wali[k]-\wali+z}\cdot\left[\frac{\sum_k \wali[k]}{\wali}-1\right]dz}
{ \int^{\wali}_{\wali(0)}\vali'(z)\frac{1}{\wali}\cdot \frac{z}{\sum_k \wali[k]-\wali+z}\cdot\left[\frac{\sum_k \wali[k]}{\sum_k \wali[k]-\wali+z}-\frac{\sum_k \wali[k]}{\wali}\right]dz}
\end{eqnarray*}
\fi
The numerator is always positive.

For the denominator, we would like to get a less complex upper bound on it by dropping the $z$ term within the brackets.  Generally we can do this but we have to be careful that the overall sign of the denominator does not change.

For $i\neq1$ and $h_i>1$, then the denominator is negative by simple inspection of the left hand side.
For $i=1$, $h_1<1$, then the denominator is positive.
We relax the denominator and increase it, arguing after the calculations that doing this does not change the sign of the expression.\ifsoda For the bound, only the bracketed term changes, so the following only includes the bracketed term of the denominator.
\begin{align*}
&\left[\frac{\sum_k \wali[k]}{\sum_k \wali[k]-\wali+z}-\frac{\sum_k \wali[k]}{\wali}\right]<\\
&
\left[\frac{\sum_k \wali[k]}{\sum_k \wali[k]-\wali}-\frac{\sum_k \wali[k]}{\wali}\right]=
\left[\frac{\left(\sum_k \wali[k]\right)\left(2\wali-\sum_k \wali[k]\right)}{\wali\left(\sum_k \wali[k]-\wali\right)}\right]
\end{align*}
\else
\begin{align*}
&\int^{\wali}_{\wali(0)}\vali'(z)\frac{1}{\wali}\cdot \frac{z}{\sum_k \wali[k]-\wali+z}\cdot\left[\frac{\sum_k \wali[k]}{\sum_k \wali[k]-\wali+z}-\frac{\sum_k \wali[k]}{\wali}\right]dz\\
&\qquad<
\int^{\wali}_{\wali(0)}\vali'(z)\frac{1}{\wali}\cdot \frac{z}{\sum_k \wali[k]-\wali+z}\cdot\left[\frac{\sum_k \wali[k]}{\sum_k \wali[k]-\wali}-\frac{\sum_k \wali[k]}{\wali}\right]dz\\
&\qquad=
\int^{\wali}_{\wali(0)}\vali'(z)\frac{1}{\wali}\cdot \frac{z}{\sum_k \wali[k]-\wali+z}\cdot\left[\frac{\left(\sum_k \wali[k]\right)\left(2\wali-\sum_k \wali[k]\right)}{\wali\left(\sum_k \wali[k]-\wali\right)}\right]dz
\end{align*}
\fi
The important term is $\left(2\wali-\sum_k \wali[k]\right)$.  For $i=1$, $\wali[1]>0.5\sum_k \wali[k]$ by Lemma~\ref{lem:halfw}, and also for $j\neq 1$, $\wali[j]<0.5\sum_k \wali[k]$ by Lemma~\ref{lem:halfw}.  Then clearly the  denominator is still positive for $i=1$; and still negative for agents $i\neq 1$.
So we give a lower bound on the fraction using the proved upper bound on the denominator.
\ifsoda
\begin{align*}
&\frac{h_i}{1-h_i}>\\
&\frac{ \int^{\wali}_{\wali(0)}\vali'(z)\frac{1}{\wali}\cdot \frac{z}{\sum_k \wali[k]-\wali+z}\cdot\left[\frac{\sum_k \wali[k]}{\wali}-1\right]dz}
{\int^{\wali}_{\wali(0)}\vali'(z)\frac{1}{\wali}\cdot \frac{z}{\sum_k \wali[k]-\wali+z}\cdot\left[\frac{\sum_k \wali[k]}{\sum_k \wali[k]-\wali}-\frac{\sum_k \wali[k]}{\wali}\right]dz}\\
&=\frac{\frac{\sum_k \wali[k]}{\wali}-1}{\frac{\sum_k \wali[k]}{\sum_k \wali[k]-\wali}-\frac{\sum_k \wali[k]}{\wali}}\\
&=\frac{(1-\alloci)^2}{2\alloci-1} \qedend
\end{align*}
\else
\begin{align*}
\frac{h_i}{1-h_i}&>
\frac{ \int^{\wali}_{\wali(0)}\vali'(z)\frac{1}{\wali}\cdot \frac{z}{\sum_k \wali[k]-\wali+z}\cdot\left[\frac{\sum_k \wali[k]}{\wali}-1\right]dz}
{\int^{\wali}_{\wali(0)}\vali'(z)\frac{1}{\wali}\cdot \frac{z}{\sum_k \wali[k]-\wali+z}\cdot\left[\frac{\sum_k \wali[k]}{\sum_k \wali[k]-\wali}-\frac{\sum_k \wali[k]}{\wali}\right]dz}\\
&=\frac{\frac{\sum_k \wali[k]}{\wali}-1}{\frac{\sum_k \wali[k]}{\sum_k \wali[k]-\wali}-\frac{\sum_k \wali[k]}{\wali}}\\
&=\frac{(1-\alloci)^2}{2\alloci-1} \qedend
\end{align*}
\fi
\end{proof}
\fi


\section{Supporting Material for \Cref{s:computation}}
\label{a:computation}

\label{a:supportingalg}

The goal of this section is to show in detail how to reduce the price
inversion question to binary search.  We do this by showing that the
analysis is largely many-to-one separable: we can make meaningful
observations about each agent individually, in particular by treating
the (initially unknown) sum total of all weights $s=\sum_k \wali[k]$
as an independent variable used as input to the analysis of each
agent.

Before getting to the key results, we use a more measured pace than is
possible in the main body of the paper to give some preliminary
analysis of the problem regarding price functions and structure of
search spaces, in particular for ``small'' agents with weight at most
half the total.  We do this in \Cref{a:b:structure} and then the rest
of this section is laid out as follows: \Cref{a:b:alg} gives both
intuition and the fully detailed version of the
algorithm; \Cref{a:b:proofs} gives the proofs of the critical lemmas
and \Cref{thm:simplealg} from \Cref{s:computation}; and finally
technical \Cref{a:b:ends} is used to support \Cref{a:b:proofs} and to describe within the algorithm how we set up ``oracle checks'' to
find the correct sub-space of weight space to search for a solution, and the endpoints of binary search.

\subsection{First Computations and Analysis of the Search Space}
\label{a:b:structure}

This section exhibits the fundamentals of a reduced, separated,
one-agent analysis of the price inversion question.  Note the following explicit conversion of the
function $\payi(\cdot)$ to accept sum $s=\sum_k\wali[k]$ as an input
variable in place of $\walsmi$.  We recall
equation~\eqref{eqn:bidfn:w}:

\begin{eqnarray*}
\label{eqn:bidfn}
\payi(\wals)&=&\vali(\wali)-\frac{\int^{\wali}_{\wali(0)} \alloci(z,\walsmi)\vali'(z)dz}{\alloci(\wals)}
\end{eqnarray*}
where we also recall $\vali(\cdot)$ is overloaded to be the function
that maps from buyer $i$'s weight back to buyer $i$'s value
(well-defined by the assumption that $\wali(\cdot)$ is strictly
increasing).  Re-arranging we have:
\ifsoda
\begin{equation}
\label{eqn:stratiofswi}
\bpayi(s,\wali)=
\vali(\wali)-\frac{s}{\wali}\int^{\wali}_{\wali(0)} \frac{z}{s-\wali+z}\vali'(z)dz
\end{equation}
\else
\begin{eqnarray}
\label{eqn:stratiofswi}
\bpayi(s,\wali)&=&
\vali(\wali)-\frac{s}{\wali}\int^{\wali}_{\wali(0)} \frac{z}{s-\wali+z}\vali'(z)dz
\end{eqnarray}
\fi

The form of equation~\eqref{eqn:stratiofswi} illustrates the critical
relationships between $\bpayi$, $s$, and $\wali$.  Our high level goal
will be to understand the behavior of the function $\bpayi$ in the
space ranging over feasible $s$ and $\wali$, starting with the
technical computations of the partials on $\bpayi$.  Recall from
\Cref{lem:pderi} that functions $\payi$ have the same cross-partials
with respect to $w_j~\forall~j\neq i$.  This property extends to
$\bpayi$:
\begin{align*}
&\bpayi(s+d\wali, \wali+d\wali)=\payi(\wali+d\wali, \walsmi)\\
&\qquad \Rightarrow \qquad \tfrac{\partial \bpayi}{\partial s}d\wali+\tfrac{\partial \bpayi}{\partial \wali}d\wali=\tfrac{\partial \payi}{\partial \wali}d\wali\\
&\bpayi(s,\wali+d\wali)=\payi(\wali+d\wali, w_j-d\wali, \wals_{-i,j})\\
&\qquad\Rightarrow \qquad \tfrac{\partial \bpayi}{\partial \wali}d\wali=\tfrac{\partial \payi}{\partial \wali}d\wali-\tfrac{\partial {\payi}}{\partial w_j}d\wali
\end{align*}
Combining the above equations together, and any $j \neq i$ we get
\begin{align}
\label{eqn:betabardir}
\frac{\partial \bpayi}{\partial \wali}&=\frac{\partial \payi}{\partial \wali}-\frac{\partial \payi}{\partial \wali[j]}\\
\frac{\partial \bpayi}{\partial s}&=\frac{\partial \payi}{\partial \wali[j]} \quad \label{eqn:betabardir2}
\end{align}

We give the intuition for these calculations.  If $\wali$ increases unilaterally without a change in $s$, then it must be that some other $w_j$ decreases by an equal amount.  If we increase $s$ without an observed change in $w_i$, then it must be some other $w_j$ that increased.\footnote{Note that because all the cross-derivatives are the same, it is without loss of generality that we assume that changes $\partial s$ are entirely attributable to one other particular agent $j\neq i$ as $\partial \wali[j]$.}  The result is the symbolic identities as given in equations~\eqref{eqn:betabardir} and~\eqref{eqn:betabardir2} above.  We will evaluate them in more detail in \Cref{thm:smallri} below.



We formally identify three objects of interest (initially discussed in
\Cref{s:computation}, see \Cref{fig:pricelevel2}).  These quantities
are defined for each agent $i$, weight function $\wali$, and the
observed price $\rezi$ of this agent.  Importantly, though the
notation includes the whole profile of observed prices $\rezs$, these
objects only depend on its $i$th coordinate $\rezi$.

\begin{itemize}
\item First, the price level set $\mathcal{Q}^{\rezs}_i$ is defined as
  $\left\{(s,\wali)\,|\,\bpayi(s,\wali)=\rezi\right\}$, i.e., these
  are the $\rezi$ level-sets of $\bpayi(s,\wali)$.  The pertinent
  subset of $\mathcal{Q}^{\rezs}_i$ is
  $\mathcal{P}^{\rezs}_i=\left\{(s,\wali)\,|\,\bpayi(s,\wali)=\rezi~\text{and}~\wali\leq s/2\right\}\subseteq\mathcal{Q}^{\rezs}_i$, i.e.,
  the subset which restricts the set $\mathcal{Q}^{\rezs}_i$ to the
  region where $\wali$ is at most half the total weight
  $s$.\footnote{We can not assume that the set $\mathcal{P}_i^{\rezs}$
    is non-empty without proof.  We prove that it is non-empty in~\Cref{thm:smallri2}.}  These sets are illustrated respectively by the dashed and
  solid lines in \Cref{fig:pricelevel2}.

\item 
  Second, the elements of the price level-set
  $\mathcal{P}_i^{\rezs}$ each have unique $s$ coordinate (see~\Cref{thm:smallri}).  It
  will be convenient to describe it as a function mapping sum $s$ to
  weight $\wali$ 
of agent $i$, parameterized by the price $\rezi$.  Denote this function
  $\zwali(s)$.  This function is illustrated in \Cref{fig:pricelevel2}
  where below the dotted line $\wali=s/2$, the curve is a
  function in $s$.  Qualitatively, it is monotone decreasing and not necessarily convex.

\item Third, $\mathcal{P}_i^{\rezs}$ is non-empty and possesses a
  smallest total weights coordinate $s$ which we define as $\minsumi = \min \{s : (s,\wali) \in
  \mathcal{P}^{\rezs}_i\}$.  In the example of \Cref{fig:pricelevel2},
  $\minsumi$ is the $s$-coordinate of the point where the level-set
  $\mathcal{P}_i^{\rezs}$ intersects the $\wali=s/2$ line.  In the
  case that the entire set $\mathcal{Q}_i^{\rezs}$ is below the $\wali
  = s/2$ line, $\mathcal{P}_i^{\rezs} = \mathcal{Q}_i^{\rezs}$ and
  $\minsumi$ is the sum $s$ that uniquely satisfies $\bpayi(s,\wali(h)) = \rezi$.%
\footnote{Further discussion will be given in \Cref{a:b:ends} where we
  show that $\minsumi$ can be computed via a binary search, between starting lower and upper bounds which are easy to find.}
\end{itemize}

\ifsodacut
\else
\begin{figure*}
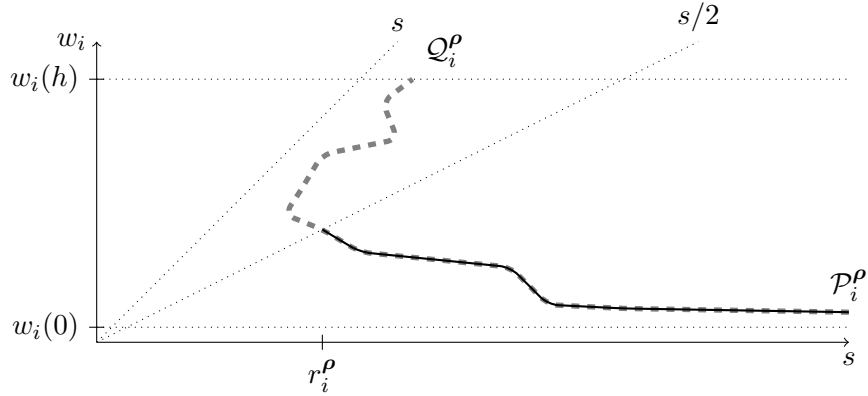

\centering
\pricelevelfigure
\caption{%
The price level set curve $\mathcal{Q}^{\rezs}_i =
  \left\{(s,\wali):\bpayi(s,\wali)=\rezi\right\}$ (thick, gray,
  dashed), is decreasing below the $\wali = s/2$ line
  (\Cref{lem:decreasingweightj}) where it is defined by its subset
  $\mathcal{P}^{\rezs}_i$ (thin, black, solid).  It is bounded above
  by the $\wali = s$ line (trivially as $s$ sums over all weights) and
  the $\wali = \wali(h)$ line (the maximum weight in the support of
  the values), and below by the $\wali = \wali(0)$ line which we have
  assumed to be strictly positive.  $\minsumi$ is the minimum weight-sum consistent
  with observed price $\rezi$ and weights $\wali \leq s/2$.  \textit{This is an exact replica  of~\Cref{fig:pricelevel}, copied here for convenience.}}
  \label{fig:pricelevel2}
\end{figure*}
\fi

Continuing, consider price level set $\mathcal{Q}^{\rezs}_i$.  We note again that $\bpayi(\cdot)$
can be used to map a domain of $(s,\wali)$ to price level sets (as
depicted in \Cref{fig:pricelevel2}).  In this context we return to
analyzing the partial derivatives of $\bpayi(\cdot)$, formally with
\Cref{thm:smallri} (immediately to follow).  Intuitively, the
statement of \Cref{thm:smallri} claims the following, with relation to
\Cref{fig:pricelevel2}:
\begin{itemize}
\item Part 1 of \Cref{thm:smallri}: below the $\wali = s/2$ line,
  starting at any point $(\hat{s},\hat{\wali})$, we strictly ``move
  up'' fixed-price level sets as we move up to
  $(\hat{s},\hat{\wali}+\delta)$, or to the right to
  $(\hat{s}+\delta,\hat{\wali})$.
\item Part 2 of \Cref{thm:smallri}: below the $\wali = s/2$ line,
  price level sets are necessarily decreasing curves; further they are
  defined for arbitrarily large $s$, which reflects the many-to-one
  nature of this analysis: other than the summary statistic $s$,
  nothing specific is known about the other agents, for example we do not need to know the
  number of other agents or their weights functions or bounds on their weights.
\item Additionally, above the $\wali = s/2$ line, we ``move up''
  fixed-price level sets with an increase in $s$ but not necessarily
  with an increase in $\wali$.
\end{itemize}

For use in \Cref{thm:smallri} and the rest of this \Cref{a:computation}, we overload the notation $h_i$ as defined in equation~\eqref{eq:hi} to be a function of $\wali$ and $s$ rather than $\wals$, with the obvious substitution in its definition to replace $\sum_k \wali[k]$ with $s$.

\begin{lemma}
\label{thm:smallri}
Assume $\wali \leq s/2$ and fix the price of agent $i$ to be
$\rezi>0$.  Let $\mathcal{Q}_i^{\rezs}$, $\mathcal{P}_i^{\rezs}$,
$\zwali(s)$ and $\minsumi$ be defined as above, and
$h_i=\frac{\partial \payi}{\partial \wali}/\frac{\partial
  \payi}{\partial \wali[j]}$ extended from equation~\eqref{eq:hi}.  Then
restricting analysis to the cone described by $\wali \leq s/2$ and
non-negative weight $\wali$:
\begin{enumerate}
\item $\bpayi(s,\wali)$ is a continuous and strictly increasing
  function in both variables $s$ and $\wali$, with specifically
  $\frac{\partial\bpayi(s,\wali)}{\partial s}=\frac{\partial
    \payi}{\partial \wali[j]}$ and $\frac{\partial \bpayi(s,\wali)}{\partial \wali} =\frac{\partial
    \payi}{\partial \wali[j]}\cdot(h_i-1)$;
\item $\zwali(s)$ is a well-defined and strictly decreasing function
  on $s \in [\minsumi,\infty)$ with $\frac{d\zwali(s)}{ds} =
    \frac{1}{1-h_i}$; in particular the function is well-defined for
    arbitrarily large $s$ independent of the number of other agents or
    their weight functions;
\item $\zwali(s)$ can be computed to arbitrary precision using binary search.
\end{enumerate}
Further, (1) partially extends such that $\bpayi(s,\wali)$ is
increasing in $s$ with $\frac{\partial\bpayi(s,\wali)}{\partial
  s}=\frac{\partial \payi}{\partial \wali[j]}$ holding everywhere,
(so including above the line $\wali = s/2$).
\end{lemma}
\begin{proof}

For (1), as in \Cref{s:oneweightspsd}, we set $h_i=\frac{\partial
  \payi}{\partial \wali}/\frac{\partial \payi}{\partial \wali[j]}$,
i.e., the diagonal entry in the Jacobian matrix after the
normalization (divide each row by its common cross-partial term), and with the substitution $s=\sum_k\wali[k]$.

From equation~\eqref{eqn:betabardir}, $\frac{\partial \bpayi}{\partial \wali}=\frac{\partial \payi}{\partial \wali}-\frac{\partial \payi}{\partial \wali[j]}=\frac{\partial \payi}{\partial \wali[j]}\cdot(h_i-1)$.
By \Cref{lem:halfw}, $h_i$ is larger than 1 when $\wali \leq s/2$. By \Cref{lem:pderi}, $\frac{\partial \payi}{\partial \wali[j]}>0$. Hence
$\frac{\partial \payi}{\partial \wali[j]}\cdot(h_i-1)>0$ when $w_i\leq s/2$.
The $\frac{\partial\bpayi}{\partial s}$ direction follows directly from equation~\eqref{eqn:betabardir2} with \Cref{lem:pderi} applying to $\frac{\partial \payi}{\partial \wali[j]}$.  This argument is also sufficient to prove the last claim of the lemma statement extending (1).

For (2), we first observe that the function $\zwali(s)$ is
well-defined (on an appropriate domain) because $\zwali(s)$ uses fixed
$\rezi$, otherwise it would contradict the monotonicity properties
proved in (1) which requires we ``move up" price level sets whenever
we unilaterally increase $\wali$.  Therefore we can take the derivative
with respect to $s$.  We get $\frac{d\zwali(s)}{ds}$ is negative for
$\wali \leq s/2$ by the following calculation (from first-order
conditions as we move along the fixed curve resulting from $\bpayi(\cdot)$
having constant output $\rezi$):
\ifsoda
\begin{align*}
&0=\frac{\partial \bpayi}{\partial \zwali(s)}d\zwali(s)+\frac{\partial \bpayi}{\partial s}ds\qquad \Rightarrow\\
& \qquad \frac{d\zwali(s)}{ds}=\frac{-\frac{\partial \bpayi}{\partial s} }{\frac{\partial \bpayi}{\partial \zwali(s)}}=
-\frac{\frac{\partial \payi}{\partial w_j}}{\frac{\partial \payi}{\partial \wali}-\frac{\partial \payi}{\partial w_j}}=\frac{1}{1-h_i}<0
\end{align*}
\else
\begin{align*}
&0=\frac{\partial \bpayi}{\partial \zwali(s)}d\zwali(s)+\frac{\partial \bpayi}{\partial s}ds\\
&\qquad \Rightarrow \qquad \frac{d\zwali(s)}{ds}=\frac{-\frac{\partial \bpayi}{\partial s} }{\frac{\partial \bpayi}{\partial \zwali(s)}}=
-\frac{\frac{\partial \payi}{\partial w_j}}{\frac{\partial \payi}{\partial \wali}-\frac{\partial \payi}{\partial w_j}}=\frac{1}{1-h_i}<0
\end{align*}
\fi

The last inequality uses \Cref{lem:halfw} from which $\wali \leq s/2$
implies $h_i > 1$.  We next prove for (2) that $\zwali(\cdot)$ and its domain are well-defined.

Technical~\Cref{thm:smallri2} (deferred to \Cref{a:b:ends}) will show that $\mathcal{P}_i^{\rezs}$ is non-empty.  Consider starting at any of its elements.  We can theoretically use its continuous derivative to ``trace out" the curve of the function $\zwali(s)$.  As $s$ increases from the starting point, we note that positive prices can never be consistent with non-positive weights, such that the continuous and negative derivative implies that the function converges to some positive infimum as $s\rightarrow\infty$.  As $s$ decreases, the function increases until either we reach a maximum feasible point with $(s,\zwali(s)=\wali(h))$ from the maximum value type $h$, or otherwise the input-output pair $(s,\zwali(s))$ intersects the line $\wali = s/2$, and minimum total weight $\minsumi$ is realized at the point of intersection.

This shows that ``reals at least $\minsumi$" is a valid domain for $\zwali(s)$, and this completes the first statement in (2).  The second statement of (2) follows because the construction of the set $\mathcal{Q}^{\rezs}_i$ is independent of other agents: for any realization of the set of other agents, their effect is summarized with the variable $s$.

For (3), we note that the output of function $\zwali(s)$ has constant lower-bound $\wali(0)$ and is upper-bounded by $s/2$, so we can indeed run binary search.
\end{proof}

Within \Cref{thm:smallri} we explicitly note the significance of the $h_i$ terms
in derivative calculations.  As the last part of the statement shows,
these derivative calculations also hold for the space $\wali > s/2$
(with a carefully extended interpretation of the $\zwali$ function to
be sure to apply the mapping from $s$ at the correct $\wali$); but we
do not get the contrapositive of \Cref{lem:halfw} in this region
to guarantee the sign of $(1-h_i)$, and so we do not get the monotonicity
property of (2) everywhere.

Further, recall the statement of~\Cref{lem:fin} (originally given on page~\pageref{lem:fin}):
\begin{numberedlemma}{\ref{lem:fin}}
When $h_1<1$ and $h_j>1~\forall j\neq1$, we have $\sum_k \frac{1}{1-h_k}>1$.
\end{numberedlemma}

As economic intuition for this result, we now see that the terms in the sum are exactly the derivatives $\frac{\partial\bpayi(s,\wali)}{\partial s}$, i.e., derivatives of the respective agents' $\rezi$ level set curves.  We will see the importance of~\Cref{lem:fin} below as the key final step in the proof of~\Cref{lem:increasingstratbar}.

\subsection{The Full Algorithm}
\label{a:b:alg}

\ifsodacut
\else
Because the observed profile of prices $\rezs$ is invertible to a
unique profile of weights (from \Cref{s:oneweightspsd}), the quantity
$s =\sum_k \wali[k]$ is uniquely determined by observed prices.  The
intuitive description of the algorithmic strategy to compute the
inversion from prices to weights is as follows.

Motivated by \Cref{a:b:structure}, we intend to split the search space for the unique $s$.  Clearly at most one of the agents can have strictly more than half the weight.
We cover the entirety of weights space by considering $n$ subspaces, representing the $n$ possibilities that any one agent $i^*$ is {\em allowed but not required} to have strictly more than half the weight.  (The region where all agents have at most half the weight is covered by all subspaces, without introducing a conflict.)  
Explicitly, define
\ifsoda
\begin{align}
\label{eqn:spacedef}
&\text{for $i\in\{1,\ldots,n\}$}\qquad \texttt{Space-i} =\\
\notag
& \left\{\wals~|~\wali\text{ unrestricted} \wedge w_j\leq \sum\nolimits_k \wali[k]/2,~\forall~j\neq i\right\}
\end{align}
\else
\begin{align}
\label{eqn:spacedef}
\texttt{Space-i} &= \left\{\wals~|~\wali\text{ unrestricted} \wedge w_j\leq \sum\nolimits_k \wali[k]/2,~\forall~j\neq i\right\}& \text{for $i\in\{1,\ldots,n\}$}
\end{align}
\fi

\noindent Recall the definition of $\zpayi[i^*]$ from line~\eqref{eqn:estbid}:
\begin{equation*}
\zpayi[i^*](s) \vcentcolon=\payi[i^*](\max\{s-\sum\nolimits_{i\neq i^*} \zwali(s),\wali[i^*](0) \}, \zwals_{-i^*}(s))
\end{equation*}

A specific monotonicity property within each \texttt{Space-i} (see
\Cref{lem:increasingstratbar} and its proof) will allow the algorithm
to use a natural binary search for the solution.  Considering such a
search in each of $n$ spaces will deterministically find $\tilde{s}$
to yield a vector of weights $\tilde{\wals}$ as
$((\tilde{s}-\sum\nolimits_{i\neq i^*}
\bwali(\tilde{s})),\zwals_{-i}(\tilde{s}))$, which are arbitrarily close
to the true $s^*$ and true $\wals^*$ (i.e.,
the $\wals^*$ which maps to $\rezs$ under $\pays$).
\fi

\ifsodacut
\else
The goal of the algorithm is to find the agent $i^*$ and unique $s$ such that $\zpayi[i^*](s)$ outputs $\rezi[i^*]$, the true payment.  I.e., we search for the equality of $\zpayi[i^*](s) = \rezi[i^*]$.

We now give the full version of the algorithm.  Beyond the outline in the main body of the paper, the most significant new technical piece in the expanded description is the use of $\minsumi[j]$ variables ($j\neq i$) to lower bound the search for $s$ in any given candidate \texttt{Space-i}.  The $\minsumi[j]$ variables were described as the third item of interest in \Cref{a:b:structure}.  They are used in the expanded descriptions of new pre-process step 0, and steps 1(a)(b)(c).  We also newly use $s(h) = \sum_k \wali[k](h)$ to denote the maximum sum of weights possible.

The full algorithm (with intuitive remarks):

\begin{enumerate}
\item[0.] \textit{Pre-process}: For each $i$, compute $\minsumi$:\footnote{See \Cref{a:b:ends} for further explanation.}
\begin{enumerate}
\item (general case: $\mathcal{P}_i^{\rezs} \neq \mathcal{Q}_i^{\rezs}$) 
   if $\bpayi(2\wali(h),\wali(h))\geq \rezi$, run binary search
  ``diagonally" on the line segment of $\wali=s/2$ between $(0,0)$ and
  $(2\wali(h),\wali(h))$ to find an element of $\mathcal{Q}_i^{\rezs}$
  and use its $s$ coordinate as $\minsumi$ (which we can do because
  $\bpayi(\cdot)$ is strictly increasing on this domain); 
\item (edge case: $\mathcal{P}_i^{\rezs}=\mathcal{Q}_i^{\rezs}$)
  otherwise, fix $\wali$ coordinate to its maximum $\wali(h)$ and run
  binary search ``horizontally" to find
  $\hat{s}\in\left[2\wali(h),s(h)\right]$ representing
  $(\hat{s},\wali(h))\in\mathcal{Q}_i^{\rezs}$ (which we can do
  because $\bpayi(\cdot)$ is strictly increasing in $s$ for constant
  $\wali$); set minimum total weight $\minsumi =\hat{s}$.
\end{enumerate}

\item find an agent $i^*$ and search a range $\left[s_L,s_H\right]$ over possible $s$ by iteratively running the following for each fixed assignment of agent $i \in \left\{1,\ldots,n\right\}$:
\begin{enumerate}
\item temporarily set $i^*=i$;
\item determine the range $\left[s_L,s_H\right]$ on which $\zpayi[i^*]$ is well-defined and searching is appropriate:
\begin{itemize}
\item identify a candidate lower bound $s_L=\max_{j\neq i^*} \minsumi[j]$
  (because any smaller $s\in\left[0,s_L\right)$ is outside the domain
  of $\zwali[j]$, for some $j$);
\item run a ``validation check'' on the lower bound, specifically, exit this iteration of the for-loop if we do not observe:
\ifsoda
\begin{align*}
 &\rezi[i^*]\geq \zpayi[i^*](s_L) =\bpayi[i^*](s_L,s_{\text{rat.}})   ~\text{for}~ s_{\text{rat.}}\\
 \notag
 & = \max\{s_L-\sum\nolimits_{k\neq i^*} \zwali[k](s_L),\wali[i^*](0)\}
\end{align*}
\else
\begin{equation*}
\zpayi[i^*](s_L) = \bpayi[i^*](s_L,\max\{s_L-\sum\nolimits_{k\neq i^*} \zwali[k](s_L),\wali[i^*](0)\}) \leq \rezi[i^*]
\end{equation*}
\fi
(because recall the goal of the algorithm, to search for equality of $\zpayi[i^*](s)=\rezi[i^*]$; but by \Cref{thm:smallri}, $\zpayi[i^*]$ is increasing,
then if the inequality here does not hold at the lower bound, the left hand side is
already too big and will never decrease);
\item identify a candidate upper bound $s_H$ using binary search to find $s_H$ as the largest total weight consistent with the maximum weight of agent $i^*$, i.e, such that $s_H-\sum_{k\neq i^*}\zwali[k](s_H)=\wali[i^*](h)$:
\begin{itemize}
\item search for $s_H\in\left[s_L,s(h)\right]$ ($s > s_H$ will ``guess'' impossible weights $\wali[i^*]>\wali[i^*](h)$ as input to $\bpayi[i^*]$, because $\wali[i^*]$ gets the balance of $s$ after subtracting the decreasing functions in $\sum_{k\neq i^*}\wali[k](s)$, see~\Cref{lem:uppersh});
\end{itemize}
\item run a ``validation check'' on the upper bound, specifically, exit this iteration of the for-loop immediately after either of the following fail (in order):
\ifsoda
\begin{align*}
\wali[i^*](0)&\leq s_H-\sum\nolimits_{k\neq i^*}\zpayi[k](s_H)\\
\rezi[i^*] &\leq \zpayi[i^*](s_H) \\
&= \bpayi[i^*](s_H,s_H-\sum\nolimits_{k\neq i^*} \zwali[k](s_H)) 
\end{align*}
\else
\begin{align*}
\wali[i^*](0)&\leq s_H-\sum\nolimits_{k\neq i^*}\zpayi[k](s_H)\\
\rezi[i^*] &\leq \zpayi[i^*](s_H) = \bpayi[i^*](s_H,s_H-\sum\nolimits_{k\neq i^*} \zwali[k](s_H)) 
\end{align*}
\fi
(with the first line checking the rationality of the interim guess of weight $\wali[i^*]$ and the second applying reasoning symmetric to the justification of the oracle on the lower bound);
\end{itemize}
\item permanently fix $i^* = i, s_L, s_H$ and break the for-loop (if this step is reached, then the range $\left[s_L,s_H\right]$ over $s$ is non-empty and in fact it definitively contains a solution by passing the checks at both end points, which is why they are ``validation" checks);
\end{enumerate}
\item use the monotonicity of $\zpayi[i^*]$ to binary search on $s$ for the true $s^*$, converging $\zpayi[i^*](s)$ to $\rezi[i^*]$;
\item when the binary search has been run to satisfactory precision and reached a final estimate $\tilde{s}$, 
output weights $\tilde{\wals}=(\tilde{s}-\sum_{i\neq i^*} \zwali(\tilde{s}),\zwalsmi[i^*](\tilde{s}))$ which invert to values $\evals$ via respective $\vali(\cdot)$ functions.
\end{enumerate}
\fi

\subsection{Proofs of \Cref{lem:decreasingweightj}, \Cref{lem:increasingstratbar}, and \Cref{thm:simplealg} (Algorithm Correctness)}
\label{a:b:proofs}

We now prove the key lemmas claimed in the main body of the paper.
The purpose of \Cref{lem:decreasingweightj} is to show that if we fix
the ``large weight candidate agent'' $i^*$ putting us in
\texttt{Space-i$^*$}, then all other agents have weights that are a
precise, monotonically decreasing function of $s$.  Critically, recall
that we can set $i^*$ to be any agent, it is not restricted to be the
agent (if any) who actually has more than half the weight (according
to the true weights of any specific problem instance).

\begin{numberedlemma}{\ref{lem:decreasingweightj}}
The price level set $\mathcal{Q}^{\rezs}_i$ is a curve; further,
restricting $\mathcal{Q}^{\rezs}_i$ to the region $\wali \leq s/2$,
the resulting subset $\mathcal{P}^{\rezs}_i$ can be written as
$\{(s,\zwali(s)):s \in [\minsumi,\infty)\}$ for a real-valued decreasing
  function $\zwali$ mapping sum $s$ to a weight $\wali$ that is
  parameterized by the observed price $\rezi$.
\end{numberedlemma}
\begin{proof}
This lemma follows as a special case of \Cref{thm:smallri}.
\end{proof}

The purpose of \Cref{lem:increasingstratbar} is to prove that function $\zpayi[i^*]$ is monotone increasing in $s$ within \texttt{Space-i$^*$}; setting up our ability to identify end points $s_L$ and $s_H$ where we run oracle checks to identify if a solution exists between them, i.e., setting up our ability to run binary search for the unique solution $s^*$ in a correct space.

\begin{numberedlemma}{\ref{lem:increasingstratbar}}
For any agent $i^*$ and $s\in [\max_{j\neq i^*}\minsumi[j],\infty)$, function $\zpayi[i^*]$ is weakly increasing; specifically, $\zpayi[i^*]$ is constant when $s-\sum\nolimits_{i\neq i^*} \zwali(s)\leq\wali[i^*](0)$ and strictly increasing otherwise.
\end{numberedlemma}

\begin{proof}
The quantity $s-\sum\nolimits_{i\neq i^*} \zwali(s)$ is monotone
increasing in $s$ as every term in the negated sum is decreasing in
$s$ (\Cref{thm:smallri}).  Therefore there are two cases: the ``guess" of weight
$\wali[i^*]^g\vcentcolon=\max\{s-\sum\nolimits_{i\neq i^*} \zwali(s),\wali[i^*](0)
\}$ lies in one of two ranges that are delineated by the threshold
where the increasing quantity $s-\sum\nolimits_{i\neq i^*} \zwali(s)$ crosses
the constant $\wali[i^*](0)$.

For small weight sums $s$ (below the threshold), the guess $\wali[i^*]$ evaluates to $\wali[i^*](0)$. In this region we have $$\zpayi[i^*](s)=\payi[i^*](\wali[i^*](0) , \zwals_{-i^*}(s))=0$$ because an agent with minimum weight (from value 0) uniquely inverts back to value of 0; and an agent with value 0 pays 0, from the definition of $\payi[i^*]$, see equation~\eqref{eq:winner-pays-bid-strats}.

The remainder of this proof is devoted to showing the second case, that the function $\zpayi[i^*](s)$
is strictly increasing for large
weight sums $s$ where the guess $\wali[i^*]^g$ for the weight of $i^*$ evaluates to
$s-\sum\nolimits_{i\neq i^*} \zwali(s)$.  For the following, we use the result of~\Cref{lem:halfw}
and the definition of $h_i$ in equation~\eqref{eq:hi}.  Note that when we are in
\texttt{Space-i$^*$}, we have $h_{k}>1$ for $k\neq i^*$.
\ifsoda
\begin{align*}
&\frac{d\zpayi[i^*](s)}{ds}=\frac{d \payi[i^*]((s-\sum\nolimits_{i\neq i^*} \zwali(s)),\zwals_{-i^*}(s))}{d s}\\
&\quad=\frac{\partial \payi[i^*]}{\partial \wali[i^*]}\left(1-\sum_{i\neq i^*}\frac{d \zwali(s)}{ds}\right)+\frac{\partial \payi[i^*]}{\partial w_{j\neq i^*}}\sum_{i\neq i^*}\frac{d \zwali(s)}{ds}\\
&\quad=\frac{\partial \payi[i^*]}{\partial \wali[i^*]}\left(1-\sum_{i\neq i^*}\frac1{1-h_i}+\frac{1}{h_{i^*}}\sum_{i\neq i^*}\frac1{1-h_i}\right)\\
&\quad=\frac{\partial \payi[i^*]}{\partial \wali[i^*]}\left[1+\left(\frac{1}{h_{i^*}}-1\right)\sum_{i\neq i^*}\frac1{1-h_i}\right]
\end{align*}
\else
\begin{eqnarray*}
\frac{d\zpayi[i^*](s)}{ds}&=&\frac{d \payi[i^*]((s-\sum\nolimits_{i\neq i^*} \zwali(s)),\zwalsmi[i^*](s))}{d s}\\
&=&\frac{\partial \payi[i^*]}{\partial \wali[i^*]}\left(1-\sum_{i\neq i^*}\frac{d \zwali(s)}{ds}\right)+\frac{\partial \payi[i^*]}{\partial w_{j\neq i^*}}\sum_{i\neq i^*}\frac{d \zwali(s)}{ds}\\
&=&\frac{\partial \payi[i^*]}{\partial \wali[i^*]}\left(1-\sum_{i\neq i^*}\frac1{1-h_i}+\frac{1}{h_{i^*}}\sum_{i\neq i^*}\frac1{1-h_i}\right)\\
&=&\frac{\partial \payi[i^*]}{\partial \wali[i^*]}\left[1+\left(\frac{1}{h_{i^*}}-1\right)\sum_{i\neq i^*}\frac1{1-h_i}\right]
\end{eqnarray*}
\fi
In the second line here, the notation $\frac{\partial \payi[i^*]}{\partial
  w_{j\neq i^*}}$ recalls that all cross-partials are the same for other agents $j$; moving
from the second line to the third line, we replaced $\frac{d
  \zwali(s)}{ds}=1/(1-h_i)$ from Part~2 of
\Cref{thm:smallri}, which also guarantees that each of these terms is strictly negative.  When $h_{i^*}\geq 1$, the total bracketed term is positive, and $\frac{d \payi[i^*]((s-\sum\nolimits_{i\neq i^*}
  \zwali(s)),\zwalsmi[i^*](s))}{d s}>0$.

Alternatively to make an argument when $h_{i^*}<1$, we further rearrange the algebra of the partial.  \ifsoda Reorganizing from the last line and restoring at the bottom:\else Continuing from the last line:\fi
\ifsoda
\begin{align*}
&\frac{1}{\partial\payi[i^*]/\partial\wali[i^*]}\cdot\frac{d\zpayi[i^*](s)}{ds}=\left[1+\left(\frac{1}{h_{i^*}}-1\right)\sum_{i\neq i^*}\frac1{1-h_i}\right]\\
&\quad=\left[\frac{\frac{1}{h_{i^*}}-1}{\frac{1}{h_{i^*}}-1}+\left(\frac{1}{h_{i^*}}-1\right)\sum_{i\neq i^*}\frac1{1-h_i}\right]\\
&\quad=\left[\left(\frac{1}{h_{i^*}}-1\right)\left(\left(\sum_{k}\frac1{1-h_k}\right)+\frac{h_{i^*}}{1-h_{i^*}}-\frac{1}{1-h_{i^*}}\right)\right]\\
&\quad=\left(\frac{1}{h_{i^*}}-1\right)\left(\sum_{k}\frac1{1-h_k}-1\right)
\intertext{such that}
&\frac{d\zpayi[i^*](s)}{ds}=\frac{\partial\payi[i^*]}{\partial\wali[i^*]}\left(\frac{1}{h_{i^*}}-1\right)\left(\sum_{k}\frac1{1-h_k}-1\right)
\end{align*}
\else
\begin{eqnarray*}
\frac{d\zpayi[i^*](s)}{ds}&=&\frac{\partial \payi[i^*]}{\partial \wali[i^*]}\left[1+\left(\frac{1}{h_{i^*}}-1\right)\sum_{i\neq i^*}\frac1{1-h_i}\right]\\
&=&\frac{\partial \payi[i^*]}{\partial \wali[i^*]}\left[\frac{\frac{1}{h_{i^*}}-1}{\frac{1}{h_{i^*}}-1}+\left(\frac{1}{h_{i^*}}-1\right)\sum_{i\neq i^*}\frac1{1-h_i}\right]\\
&=&\frac{\partial \payi[i^*]}{\partial \wali[i^*]}\left[\left(\frac{1}{h_{i^*}}-1\right)\left(\left(\sum_{i\neq i^*}\frac1{1-h_i}\right)+\frac{h_{i^*}}{1-h_{i^*}}+\frac{1}{1-h_{i^*}}-\frac{1}{1-h_{i^*}}\right)\right]\\
&=&\frac{\partial \payi[i^*]}{\partial \wali[i^*]}\left[\left(\frac{1}{h_{i^*}}-1\right)\left(\left(\sum_{k}\frac1{1-h_k}\right)+\frac{h_{i^*}}{1-h_{i^*}}-\frac{1}{1-h_{i^*}}\right)\right]\\
&=&\frac{\partial \payi[i^*]}{\partial \wali[i^*]}\left(\frac{1}{h_{i^*}}-1\right)\left(\sum_{k}\frac1{1-h_k}-1\right)
\end{eqnarray*}
\fi
When $h_{i^*}<1$, this quantity is again necessarily positive (with the last term positive by Lemma \ref{lem:fin}).\footnote{The significance of~\Cref{lem:fin} here was discussed after the original statement of~\Cref{lem:increasingstratbar} at the end of ~\Cref{s:compsum}, and after~\Cref{thm:smallri} at the end of~\Cref{a:b:structure}.}  So again $\frac{d \payi[i^*]((s-\sum\nolimits_{i\neq i^*} \bwali(s)),\zwalsmi[i^*](s))}{d s}>0$.  We conclude that $\zpayi[i^*](s)\vcentcolon=\payi[i^*]((s-\sum\nolimits_{i\neq i^*} \zwali(s)),\zwalsmi[i^*](s))$ is strictly increasing in $s$ for this second case, i.e., the region of large $s$ for which \begin{equation*}\wali[i^*]^g\vcentcolon=\max\{s-\sum\nolimits_{i\neq i^*} \zwali(s),\wali[i^*](0)
\} = s-\sum\nolimits_{i\neq i^*} \zwali(s)\qedend\end{equation*}
\end{proof}

Finally we argue the correctness of the algorithm.  However,
correctness of the technical computations in pre-processing step 0
will be delayed to \Cref{a:b:ends}.


\begin{numberedtheorem}{\ref{thm:simplealg}}
Given weights $\wals$ and payments $\rezs = \pays(\wals)$ according to
a proportional weights social choice function, the algorithm
identifies weights $\tilde{\wals}$ to within $\epsilon$ of the true
weights $\wals$ in time polynomial in the number of agents $n$, the
logarithm of the ratio of high to low weights $\max\nolimits_i\ln
(\wali(h)/\wali(0))$, and the logarithm of the desired precision $\ln
1/\epsilon$.  
\end{numberedtheorem}

\begin{proof}

Fix observed prices $\rezs$ that correspond to true weights $\wals$
with sum $s = \sum_i \wali$.  Fix an agent $i^*$ with $\wali[i^*] >
s/2$ if one exists or $i^* = 1$ if none exists.  Set $s_L =
\max_{i\neq i^*} \minsumi$, and $s_H$ as calculated in the algorithm for
\texttt{Space-i$^*$}.  It must be that $\zpayi[i^*](s_L) \leq
\rezi[i^*] \leq \zpayi[i^*](s_H)$.  The bounds follow by $\wali
\leq s/2$ for all $i\neq i^*$, and \Cref{thm:smallri2} and \Cref{lem:uppersh} (stated and proved in the next
section).  
Monotonicity of $\zpayi[i^*](\cdot)$ then implies
binary search will identify a sum $\tilde{s}$ arbitrarily close to satisfying
$\zpayi[i^*](\tilde{s}) = \rezi[i^*]$.

By the definition of $\zpayi(\cdot)$ and the convergence $\tilde{s}\rightarrow s$,
the weights $\tilde{\wals}=\wals^{\rezs}(\tilde{s})$ satisfy $\pays(\tilde{\wals})
\approx \rezs$.  We discuss rates of convergence below, but this follows because $\zwali[i\neq i^*]$ functions are decreasing in input $\tilde{s}$, so as the range of possible total weight decreases, their range of output also decreases (while still containing the solution).  The range of the guess $\zwali[i^*]$ for the weight of $i^*$ is upper-bounded by a simple additive function of the ranges of possible $s$ and $\walsmi$ (see \Cref{lem:sandwktowi}), so it is also decreasing with each binary search iteration.  By uniqueness of the inverse $\pays^{-1}$, these weights
are converging to the original weights, i.e.,  $\tilde{\wals} \approx \wals$.


In the case where $\wali[i^*] > s/2$, the iterative searches of \texttt{Space-i}
for $i\neq i^*$ will fail as these searches only consider points
$(s,\wali[i^*])$ where $\wali[i^*] < s/2$, but the weights $\wals$ that
corresponds to $\rezs$ are unique (by~\Cref{thm:themainresult}) and do not satisfy $\wali[i^*] < s/2$.
When $\wali[i^*] \leq s$ then all searches, in particular $i^*=1$, will converge to the same result of $\wals$.

Lastly, we show that binary search over $s$-coordinates within \texttt{Space-$i^*$} is sufficient to converge the algorithm's approximate $\tilde{\wals}$ to $\wals$ (measured by $\mathcal{L}_1$-norm distance) at the same assymptotic rate of the binary search on $s$, a rate which has only polynomial dependence on $n$, $\max\nolimits_i \ln(\wali(h)/\wali(0))$, and $\ln 1/\epsilon$.

By \Cref{lem:slopebounds} below, for each agent $k\neq i^*$ there is a bound $B_k$ on the magnitude of the slope of $\frac{\partial \zpayi}{\partial s}$ as a function of the value space and weight functions inputs to the problem.  $B_k$ depends on the factor $\wali/\wali(0)\leq \wali(h)/\wali(0)$ leading to the running time dependence.

Given a binary-search-step range on $s$ with size $S$, for every agent $k\neq i^*$, the size of the range containing $\wali[k]$ can not be larger than $s\cdot B_k$.  Every time the range of $s$ gets cut in half, this upper bound on the range of $\wali[k]$ also gets cut in half.  The convergence of $\tilde{\wal}_{i^*}$ to $\wali[i^*]$ follows from the convergence in coordinates $s, \walsmi[i^*]$ and \Cref{lem:sandwktowi}.
\end{proof}

\ifsodacut
\else
We conclude this section with the lemmas supporting the convergence rate claims of~\Cref{thm:simplealg}.  Within the statement of~\Cref{lem:slopebounds} recall that the definition of the derivative was proved by \Cref{thm:smallri}.
\fi

\begin{lemma}
\label{lem:slopebounds}
Given agent $i$ with $\wali\leq s/2$ and function $\zpayi$, the slope $\frac{\partial \zpayi}{\partial s} =\frac{1}{1-h_i}<0$ has magnitude bounded by $\frac{\wali}{2\wali(0)}\leq \frac{\wali(h)}{2\wali(0)}$.
\end{lemma}
\ifsodacut
\else
\begin{proof}
We will show $\big| \frac{1}{1-h_i} \big| \leq
\frac{\wali}{2\wali(0)}$.  To upper bound $\big| \frac{1}{1-h_i}
\big|$, we lower bound $h_i>1$.  Note that a lower bound on $h_i$ will
only be useful for us if it strictly separates $h_i$ above 1.
Substitute $s = \sum_k \wali[k]$ into the definition of $h_i$ in
equation~\eqref{eq:hi} and bound, with justification to follow, as:

\begin{align*}
h_i
&=\frac{ \int^{\wali}_{\wali(0)}\vali'(z)\frac{1}{\wali}\cdot \frac{z}{s-\wali+z}\cdot\left[\frac{s}{\wali}-1\right]dz}{ \int^{\wali}_{\wali(0)}\vali'(z)\frac{1}{\wali}\cdot \frac{z}{s-\wali+z}\cdot\left[\frac{s}{s-\wali+z}-1\right]dz}\\
%
%
&\geq\frac{\left[\frac{s}{\wali}-1\right]} {\left[\frac{s}{s-\wali+\wali(0)}-1\right]} \cdot \frac{ \int^{\wali}_{\wali(0)}\vali'(z)\frac{1}{\wali}\cdot \frac{z}{s-\wali+z}dz}{ \int^{\wali}_{\wali(0)}\vali'(z)\frac{1}{\wali}\cdot \frac{z}{s-\wali+z}dz}\\
%
&\geq
\frac{s-\wali+\wali(0) }{\wali-\wali(0)}
%
%
%
> 1+\frac{2\wali(0)}{\wali} > 1
\end{align*}
The first inequality replaces the integrand $z$ in the bracketted term
in the denominator with its constant lower bound $\wali(0)$ (which
only decreases a denominator, in {\em the} denominator); thereafter both bracketted terms can be
brought outside of their respective integrals.  The second inequality replaces the
numerator with 1 because $\wali\leq s/2$ by statement assumption.  The third (strict) inequality both replaces $s$ with $2\wali$ by the
same reason, and adds $\wali(0)$ to both numerator
and denominator, which makes the fraction smaller because it was
originally larger than 1.

Using this bound we get:
\begin{equation*}
\bigg| \frac{1}{1-h_i} \bigg|\leq \Bigg| \frac{1}{1-\left(1+\frac{2\wali(0)}{\wali}\right)}\Bigg| =\frac{\wali}{2\wali(0)}\leq \frac{\wali(h)}{2\wali(0)}\qedend
\end{equation*} \qedend
\end{proof}
\fi

\begin{lemma}
\label{lem:sandwktowi}
Given agent $i^*$, true $s^*\in\left[s^-,s^+\right]$, and true weights $\wali[k]\in\left[\wali[k]^-,\wali[k]^+\right]$ for agents $k\neq i^*$, which induce the range for $i^*$'s weight of $\wali[i^*]\in\left[s^--\sum_{k\neq i^*} \wali[k]^+,s^+-\sum_{k\neq i^*} \wali[k]^-\right]$.  If the sizes of the ranges $\left[s^-,s^+\right]$ and $\left[\wali[k]^-,\wali[k]^+\right]$ are each individually reduced by (at least) a constant factor $\alpha$, then the size of the range of $\wali[i^*]$ is also reduced by (at least) $\alpha$.
\end{lemma}
\ifsodacut
\else
\begin{proof}
The statement follows immediately from the induced range of $\wali[i^*]$.  Its size is exactly equal to the sum of the $n$ other ranges, i.e.,
\ifsoda
\begin{align*}
    &\bigg| \left[s^--\sum\nolimits_{k\neq i^*} \wali[k]^+,s^+-\sum\nolimits_{k\neq i^*} \wali[k]^-\right]\bigg|\\
    &=\bigg| \left[s^-,s^+\right]\bigg| + \sum_{k\neq i^*}\bigg| \left[\wali[k]^-,\wali[k]^+\right] \bigg| \qedend
\end{align*}
\else
\begin{equation*}
    \bigg| \left[s^--\sum\nolimits_{k\neq i^*} \wali[k]^+,s^+-\sum\nolimits_{k\neq i^*} \wali[k]^-\right]\bigg| = \bigg| \left[s^-,s^+\right]\bigg| + \sum_{k\neq i^*}\bigg| \left[\wali[k]^-,\wali[k]^+\right] \bigg| \qedend
\end{equation*}
\fi
\end{proof}
\fi

\subsection{Correctness of Algorithm Search End Points as Oracle Checks}
\label{a:b:ends}

This section has four purposes:
\begin{itemize}
\item analyze the structure of $\minsumi$ corresponding to level set
  $\mathcal{Q}_i^{\rezs}$;
\item prove the correctness and run-time of the pre-process step 0 of
  the algorithm, which pre-computes $\minsumi$ for all $i$;
\item conclude that the lower bounds $s_L$ of search in any given
  \texttt{Space-i}, determined within each iteration of step 1 of the
  algorithm, are the correct lower bounds of feasibility;
  \item conclude that the upper bounds $s_H$ calculated within each iteration of step 1 are the correct upper bounds of feasibility.
\end{itemize}
\ifsodacut
\else
For strictly positive observed payment $\rezi > 0$, the level set
$\mathcal{Q}_i^{\rezs}$ takes on the full range of weights $\wali \in
(\wali(\rezi),\wali(h)]$ (the lower bound of $\wali(\rezi)$ will not
play an important role, our algorithms will use the less restrictive
bound of $\wali(0)$ instead).  Our search for the minimum
$s$-coordinate of $\mathcal{P}_i^{\rezs}$, i.e., $\minsumi$, which is
the intersection of $\mathcal{Q}_i^{\rezs}$ with the points below the
$\wali = s/2$ line is either on the $\wali=s/2$ boundary or on the
$\wali = \wali(h)$ boundary.  This follows because constrained to
$\wali \leq s/2$ the level set is given by a decreasing function
(\Cref{lem:decreasingweightj}) and all level sets extend to $s =
\infty$ (this second fact is true, but will not need to be explicitly
proven).  The two cases are depicted in \Cref{fig:pricelevel3}.  For
convenience, we restate the preprocessing step of the algorithm:

\begin{figure*}
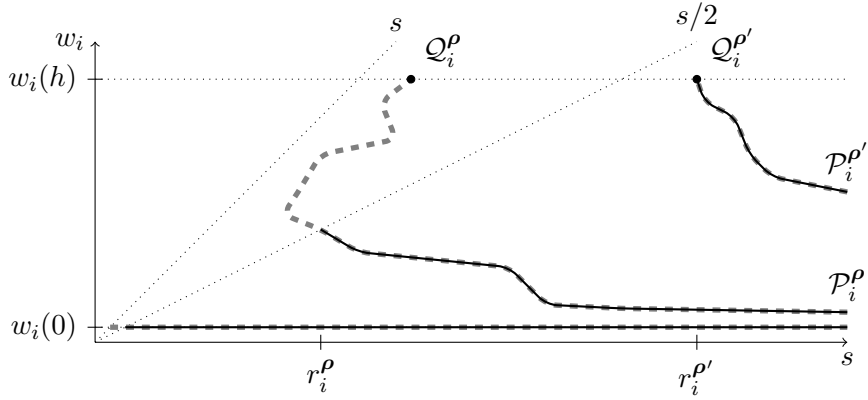

\centering
\pricelevelcasesfigure
\caption{The cases for the initialization of lower bound $\minsumi$
  are depicted.  When $\rezi=0$ both of the corresponding price level sets ${\cal P}_i$ and ${\cal Q}_i$ are on the line
  $\wali = \wali(0)$ (depicted, but not labeled). For observed price
  $\rezi \leq \bpayi(2\wali(h),\wali(h))$ the intermediate level sets
  look like the depicted $\mathcal{P}_i^{\rezs} \neq
  \mathcal{Q}_i^{\rezs}$, and $\minsumi$ corresponds to the
  $s$-coordinate at the intersection with the $\wali = s/2$ line. For
  observed price $\rezi \geq \bpayi(2\wali(h),\wali(h))$ the high
  level sets look like the depicted $\mathcal{P}_i^{\rezs'} = \mathcal{Q}_i^{\rezs'}$, and $\minsumi{}^{{}'}$ corresponds to the
  $s$-coordinate at the intersection with the $\wali = \wali(h)$
  line.}
  \label{fig:pricelevel3}
\end{figure*}

\begin{enumerate}
\item[0.] \textit{Pre-process}: For each $i$, compute $\minsumi$:
\begin{enumerate}
\item (general case: $\mathcal{P}_i^{\rezs} \neq \mathcal{Q}_i^{\rezs}$) 
   if $\bpayi(2\wali(h),\wali(h))\geq \rezi$, run binary search
  ``diagonally" on the line segment of $\wali=s/2$ between $(0,0)$ and
  $(2\wali(h),\wali(h))$ to find an element of $\mathcal{Q}_i^{\rezs}$
  and use its $s$ coordinate as $\minsumi$ (which we can do because
  $\bpayi(\cdot)$ is strictly increasing on this domain); 
\item (edge case: $\mathcal{P}_i^{\rezs}=\mathcal{Q}_i^{\rezs}$) otherwise, fix $\wali$ coordinate to its maximum $\wali(h)$ and
  run binary search ``horizontally" to find
  $\hat{s}\in\left[2\wali(h),s(h)\right]$ representing
  $(\hat{s},\wali(h))\in\mathcal{Q}_i^{\rezs}$ (which we can do
  because $\bpayi(\cdot)$ is strictly increasing in $s$ for constant
  $\wali$); set minimum total weight $\minsumi
  =\hat{s}$.
\end{enumerate}
\end{enumerate}

There is an intuitive explanation to the order of operations in the
pre-processing step 0.  First we check if we are in the general case.
We can do this because price level-sets are strictly increasing on the
line $\wali=s/2$ (see \Cref{thm:smallri2} below, extending
\Cref{thm:smallri}).  So we can check the largest the price at the
largest possible point as $\bpayi(2\wali(h), \wali(h))$; if it is too
big, we can run binary search down to $\bpayi(2\wali(0),\wali(0))=0$;
otherwise we are in the edge case where $\minsumi$ corresponds to
$\wali(h)$.  In this case, we can binary search the line $\wali =
\wali(h)$ for the point with payment $\rezi$ as, again, price level-sets are
strictly increasing (\Cref{thm:smallri}).  The formal proof is given
as \Cref{thm:smallri2}.
\fi

\begin{lemma}
\label{thm:smallri2}
For any realizable payment $\rezi$, price level set
$\mathcal{P}_i^{\rezs}$ is non-empty and its $s$-coordinates are lower
bounded by $\minsumi$ which can be computed to arbitrary precision by a
binary search.
\end{lemma}
\ifsodacut
\else
\begin{proof}  
As mentioned previously, denote the maximum sum of weights possible by $s(h) = \sum_i
\wali(h)$.  To find $\minsumi$, we first focus attention on the horizontal line with
constant weight $\wali(h)$.

A point $(\hat{s},\wali(h))$ on price
level set $\mathcal{Q}_i^{\rezs}$, i.e., with
$\bpayi(\hat{s},\wali(h)) = \rezi$, can be found to arbitrary
precision with binary search over $s \in (\wali(h),s(h)]$.
Correctness of this binary search follows because a realizable payment
$\rezi$ must satisfy $0 = \bpayi(\wali(h),\wali(h)) \leq \rezi \leq
\bpayi(s(h),\wali(h))$ and because increasing $s$-coordinate
corresponds to increasing price-level set on any line with fixed
weight $\wali$ by \Cref{thm:smallri}.  For the lower bound on the
range, an agent wins with certainty and makes no payment when the sum
of the other agent weights is zero; the upper bound is from the
natural upper bound $s \leq s(h)$.

There are now two cases depending on whether this point
$(\hat{s},\wali(h))$ is above or below the $\wali = s/2$ line.\footnote{If on the line, the cases are equal and either suffices.}  If below, then $\minsumi = \hat{s}$ because this point is tight to the maximum weight $\wali(h)$ (see~\Cref{fig:pricelevel3}), and (again by \Cref{thm:smallri}) the slope of curve $\mathcal{P}_i^{\rezs}$ is strictly negative and all smaller $s$ are infeasible.

Alternatively suppose $(\wali(h),\hat{s})$ is above the $\wali = s/2$ line, then $\minsumi$ can be found by searching the $\wali = s/2$ line.  Part (1) of \Cref{thm:smallri} guarantees that points on this line are
consistent with unique and increasing observed prices (partials of the
price function are strictly positive in both dimensions $\wali$ and
$s$, we can first move right $ds$, and then move up $d\wali$, with the
price function strictly increasing as a result of both ``moves'').  On
this line we have $0 = \bpayi(2\wali(0),\wali(0)) \leq \rezi \leq
\bpayi(2\wali(h),\wali(h))$ where the lower bound observes an agent with value 0 to always pay 0, and the upper bound follows from the supposition $\wali(h) \geq \hat{s}/2$ of this case.  Thus, a binary
search of the $\wali = s/2$ line with $\wali \in [\wali(0),\wali(h)]$
is guaranteed to find a point with price arbitrarily close to $\rezi$.
Since $\mathcal{P}_i^{\rezs}$ as a curve is decreasing in $s$, the
identified point, which is in $\mathcal{P}_i^{\rezs}$, has the minimum
$s$-coordinate.

The two cases are exhaustive and so $\minsumi$ is identified and $\mathcal{P}_i^{\rezs}$ is non-empty.
\end{proof}

\noindent We finish the section with the lemma showing the correctness of the
search range of sum $s$ within $[s_L,s_H]$.
\fi

\begin{lemma}
\label{lem:uppersh}
For true weights $\wals$, true weight sum $s = \sum \wali$, and $i^*$ with $\wali \leq s/2$ for $i \neq i^*$, sum $s$ is contained in interval $[s_L,s_H]$ (defined in step~1 of the algorithm for $i^*$).
\end{lemma}
\ifsodacut
\else
\begin{proof}
First for the lower bound $s_L$, the assumption of the lemma requires $i
\neq i^*$ satisfy $\wali \leq s/2$.  Therefore the true pair $(s,\wali)$
must be a point in ${\mathcal P}_i^{\rezs}$, and the true sum $s$
must be at least the lower bound $\minsumi$ for each $i \neq i^*$.

Second for the upper bound $s_H$, recall the definition 
\begin{equation*}
\zpayi[i^*](s) \vcentcolon=\payi[i^*](\max\{s-\sum\nolimits_{i\neq i^*} \zwali(s),\wali[i^*](0)\}, \zwals_{-i^*}(s))
\end{equation*}
which uses a guess at the total weights $\tilde{s}$ to guess the
corresponding the weight of agent $i^*$ as $(\tilde{s}-\sum\nolimits_{i\neq
  i^*} \zwali(\tilde{s}))$.  In fact, this guessed weight is strictly
increasing in $\tilde{s}$ as each term in the negated sum is strictly
decreasing (\Cref{thm:smallri}).  Our choice of $s_H$ equates this
guessed weight with its highest possible value $\wali[i^*](h)$.  By
monotonicity of the guessed weight the true $s$ must be at most $s_H$.
\end{proof}
\fi

\end{appendix}

\end{document}